\documentclass[10pt]{article}

\title{Joint-sparse recovery from multiple measurements%
  \thanks{Department of Computer Science, University of British
    Columbia, Vancouver V6T 1Z4, BC, Canada ({\tt
      \{ewout78,mpf\}@cs.ubc.ca}). Research partially supported by the
    Natural Sciences and Engineering Research Council of Canada.}  }
\author{Ewout van den Berg \and Michael P. Friedlander}

\newcommand{\trnumber}{TR-2009-07} \newcommand{\trdate}{April 2009}

\usepackage[left=1.3in,right=1.3in,
            top=1in,bottom=1in,
            nohead,headheight=28pt]{geometry}
\usepackage[small]{titlesec}
\usepackage{amsmath,amssymb,amsthm}
\usepackage{color}
\usepackage{graphicx}
\usepackage{subfigure}
\usepackage{url}
\usepackage{microtype}
\usepackage{braket}
\usepackage{cite}

\numberwithin{equation}{section}
\newcommand{\iprod}[2]{\Braket{#1,#2}}
\newcommand{\rembo}{ReMBo}

\newcommand{\spark}{\ensuremath{\textrm{Spark}}}   
\newcommand{\support}{\ensuremath{\textrm{Supp}}}  
\newcommand{\kernel}{\ensuremath{\textrm{Ker}}}    
\newcommand{\card}[1]{\ensuremath{\abs{#1}}} 
\newcommand{\col}[1]{\ensuremath{^{\scalebox{.6}{$\downarrow$}{#1}}}} 
\newcommand{\row}[1]{\ensuremath{^{{#1}{\scalebox{.6}{$\rightarrow$}}}}} 
\newcommand{\polar}{\ensuremath{^{\circ}}}         

\newtheorem{theorem}{Theorem}[section]

\newtheorem{corollary}[theorem]{Corollary}

\input{macros}

\usepackage[vlined,ruled,nofillcomment]{algorithm2e}
\SetAlgoSkip{}\SetAlgoInsideSkip{smallskip}

\usepackage{calc}
\usepackage{fancyhdr}
\fancypagestyle{plain}{%
\fancyhf{}%
\fancyhead[C]{%
  \fbox{%
    \begin{minipage}{\linewidth-2\fboxsep-2\fboxrule}
      \centering\sffamily\bfseries
      Department of Computer Science, University of British Columbia
    \\Technical Report \trnumber, \trdate
    \end{minipage}}}
\fancyfoot[C]{\thepage}
}

\begin{document}

\maketitle

\begin{abstract}
  The joint-sparse recovery problem aims to recover, from sets of
  compressed measurements, unknown sparse matrices with nonzero
  entries restricted to a subset of rows. This is an extension of the
  single-measurement-vector (SMV) problem widely studied in compressed
  sensing. We analyze the recovery properties for two types of
  recovery algorithms. First, we show that recovery using sum-of-norm
  minimization cannot exceed the uniform recovery rate of sequential
  SMV using $\ell_1$ minimization, and that there are problems that
  can be solved with one approach but not with the other. Second, we
  analyze the performance of the \rembo{} algorithm [M.\@ Mishali and
  Y. Eldar, {\it IEEE Trans.\@ Sig.\@ Proc.}, 56 (2008)] in
  combination with $\ell_1$ minimization, and show how recovery
  improves as more measurements are taken. From this analysis it
  follows that having more measurements than number of nonzero rows
  does not improve the potential theoretical recovery rate.
\end{abstract}

\section{Introduction}

A problem of central importance in compressed sensing
\cite{CAN2006a,DON2006c} is the following: given an $m\times n$ matrix
$A$, and a measurement vector $b = Ax_0$, recover $x_0$. When $m < n$,
this problem is ill-posed, and it is not generally possible to
uniquely recover $x_0$ without some prior information. In many
important cases, $x_0$ is known to be sparse, and it may be
appropriate to solve
\begin{equation}\label{Eq:L0}
\minimize{x\in\Real^n}\quad \norm{x}_0\quad \st\quad Ax = b,
\end{equation}
to find the sparsest possible solution. (The $\ell_0$-norm
$\norm{\cdot}_0$ of a vector counts the number of nonzero entries.)
If $x_0$ has fewer than $s/2$ nonzero entries, where $s$ is the number
of nonzeros in the sparsest null-vector of $A$, then $x_0$ is the
unique solution of this optimization
problem~\cite{DON2003Ea,GRI2003Na}. The main obstacle of this approach
is that it is combinatorial~\cite{NAT1995a}, and therefore impractical
for all but the smallest problems. To overcome this, Chen et
al.~\cite{CHE1998DSa} introduced basis pursuit:
\begin{equation}\label{Eq:L1}
\minimize{x\in\Real^n}\quad \norm{x}_1\quad \st\quad Ax = b.
\end{equation}
This convex relaxation, based on the $\ell_1$-norm $\norm{x}_1$, can
be solved much more efficiently; moreover, under certain conditions
\cite{CAN2006RTa,DON2005a}, it yields the same solution as the
$\ell_0$ problem~\eqref{Eq:L0}.

A natural extension of the single-measurement-vector (SMV) problem
just described is the multiple-measurement-vector (MMV)
problem. Instead of a single measurement $b$, we are given a set of
$r$ measurements $$b^{(k)} = Ax_0^{(k)}, \quad k=1,\ldots,r,$$ in
which the vectors $x_0^{(k)}$ are jointly sparse---i.e., have nonzero
entries at the same locations.  Such problems arise in source
localization~\cite{MAL2005CWb}, neuromagnetic
imaging~\cite{COT2005REKa}, and equalization of sparse-communication
channels~\cite{COT2002Ra,FEV1999GFa}.  Succinctly, the aim of the MMV
problem is to recover $X_0$ from observations $B=AX_0$, where
$B=[b^{(1)},\ b^{(2)}, \ldots,\ b^{(r)}]$ is an $m\times r$ matrix,
and the $n\times r$ matrix $X_0$ is row sparse---i.e., it has nonzero
entries in only a small number of rows.  The most widely studied
approach to the MMV problem is based on solving the convex
optimization problem
\begin{equation*}\label{Eq:L0p}
  \minimize{X\in\Real^{n\times r}}\quad\norm{X}_{p,q}\quad\st\quad AX = B,
\end{equation*}
where the mixed $\ell_{p,q}$ norm of $X$ is defined as
\[
\norm{X}_{p,q} = \Big(\sum_{j=1}^n \norm{X\row{j}}_q^p\Big)^{1/p},
\]
and $X\row{j}$ is the (column) vector whose entries form the $j$th row
of $X$. In particular, Cotter et al.~\cite{COT2005REKa} consider
$p=2$, $q\leq 1$; Tropp~\cite{TRO2006b,TRO2006GSa} analyzes $p=1$, $q
= \infty$; Malioutov et al.~\cite{MAL2005CWb} and Eldar and Mishali
\cite{ELD2008Ma} use $p=1$, $q = 2$; and Chen and Huo \cite{CHE2006Ha}
study $p=1$, $q \geq 1$. A different approach is given by Mishali and
Eldar~\cite{MIS2008Ea}, who propose the \rembo\ algorithm, which
reduces MMV to a series of SMV problems.

In this paper we study the sum-of-norms problem and the conditions
for uniform recovery of all $X_0$ with a fixed row support, and
compare this against recovery using $\ell_{1,1}$. We then construct
matrices $X_0$ that cannot be recovered using $\ell_{1,1}$ but for
which $\ell_{1,2}$ does succeed, and vice versa. We then illustrate
the individual recovery properties of $\ell_{1,1}$ and $\ell_{1,2}$
with empirical results. We further show how recovery via $\ell_{1,1}$
changes as the number of measurements increases, and propose a
boosted-$\ell_1$ approach to improve on the $\ell_{1,1}$
approach. This analysis provides the starting point for our study of
the recovery properties of \rembo, based on a geometrical
interpretation of this algorithm.

We begin in Section~\ref{Sec:L1} by summarizing existing
$\ell_0$-$\ell_1$ equivalence results, which give conditions under
which the solution of the $\ell_1$ relaxation~\eqref{Eq:L1} coincides
with the solution of the $\ell_0$ problem~\eqref{Eq:L0}. In
Section~\ref{Sec:SumOfRowNorms} we consider the $\ell_{1,2}$
mixed-norm and sum-of-norms formulations and compare their performance
against $\ell_{1,1}$. In Sections~\ref{Sec:BoostedL1}
and~\ref{Sec:ReMBo} we examine two approaches that are based on
sequential application of~\eqref{Eq:L1}.

\paragraph{Notation.}
We assume throughout that $A$ is a full-rank matrix in $\Real^{m\times
  n}$, and that $X_0$ is an $s$ row-sparse matrix in $\Real^{n\times
  r}$.  We follow the convention that all vectors are column
vectors. For an arbitrary matrix $M$, its $j$th column is denoted by
the column vector $M\col{j}$; its $i$th row is the transpose of the
column vector $M\row{i}$. The $i$th entry of a vector $v$ is denoted
by $v_i$. We make exceptions for $e_i = I\col{i}$ and for $x_0$
(resp., $X_0$), which represents the sparse vector (resp., matrix) we
want to recover. When there is no ambiguity we sometimes write $m_i$
to denote $M\col{i}$. When concatenating vectors into matrices,
$[a,b,c]$ denotes horizontal concatenation and $[a;b;c]$ denotes
vertical concatenation. When indexing with $\Iscr$, we define the
vector $v_{\Iscr} := [v_i]_{i\in\Iscr}$, and the $m\times\card{\Iscr}$
matrix $A_{\Iscr} := [A\col{j}]_{j\in\Iscr}$. Row or column selection
takes precedence over all other operators. 

\section{Existing results for $\ell_1$ recovery}\label{Sec:L1}

The conditions under which \eqref{Eq:L1} gives the sparsest possible
solution have been studied by applying a number of different
techniques. By far the most popular analytical approach is based on
the restricted isometry property, introduced by Cand\`es and Tao
\cite{CAN2005Tb}, which gives sufficient conditions for
equivalence. Donoho \cite{DON2004b} obtains necessary and sufficient
(NS) conditions by analyzing the underlying geometry of
\eqref{Eq:L1}. Several authors \cite{DON2001Ha,GRI2003Na,DON2003Ea}
characterize the NS-conditions in terms of properties of the kernel of
$A$:
\[
\kernel(A) = \{x \mid Ax=0\}.
\]
Fuchs \cite{FUC2004a} and Tropp \cite{TRO2005a} express
sufficient conditions in terms of the solution of the dual of
\eqref{Eq:L1}:
\begin{equation}\label{Eq:L1Dual}
\maximize{y}\quad b\T y\quad\st\quad \norm{A\T y}_{\infty} \leq 1.
\end{equation}
In this paper we are mainly concerned with the geometric and kernel
conditions. We use the geometrical interpretation of the problems
to get a better understanding, and resort to the null-space properties
of $A$ to analyze recovery. To make the discussion more
self-contained, we briefly recall some of the relevant results in the
next three sections.

\subsection{The geometry of $\ell_1$ recovery}

The set of all points of the unit $\ell_1$-ball, $\{x\in\Real^n \mid
\norm{x}_1 \leq 1\}$, can be formed by taking convex combinations of
$\pm e_j$, the signed columns of the identity matrix. Geometrically
this is equivalent to taking the convex hull of these vectors, giving
the cross-polytope $\Cscr = \mathrm{conv}\{\pm e_1,\pm e_2,\ldots,\pm
e_n\}$. Likewise, we can look at the linear mapping $x\mapsto Ax$ for
all points $x\in\Cscr$, giving the polytope $\Pscr = \{ Ax \mid
x\in\Cscr\} = A\Cscr$. The faces of $\Cscr$ can be expressed as the
convex hull of subsets of vertices, not including pairs that are
reflections with respect to the origin (such pairs are sometimes
erroneously referred to as antipodal, which is a slightly more general
concept~\cite{GRU2003a}). Under linear transformations, each face from
the cross-polytope $\Cscr$ either maps to a face on $\Pscr$ or
vanishes into the interior of $\Pscr$.

The solution found by \eqref{Eq:L1} can be interpreted as
follows. Starting with a radius of zero, we slowly ``inflate''
$\mathcal{P}$ until it first touches $b$. The radius at which this
happens corresponds to the $\ell_1$-norm of the solution $x^*$. The
vertices whose convex hull is the face touching $b$ determine the
location and sign of the non-zero entries of $x^*$, while the position
where $b$ touches the face determines their relative
weights. Donoho~\cite{DON2004b} shows that $x_0$ can be recovered from
$b = Ax_0$ using \eqref{Eq:L1} if and only if the face of the (scaled)
cross-polytope containing $x_0$ maps to a face on $\mathcal{P}$. Two
direct consequences are that recovery depends only on the sign pattern
of $x_0$, and that the probability of recovering a random $s$-sparse
vector is equal to the ratio of the number of $(s-1)$-faces in
$\mathcal{P}$ to the number of $(s-1)$-faces in $\Cscr$. That
is, letting $\mathcal{F}_d(\mathcal{P})$ denote the collection of all
$d$-faces \cite{GRU2003a} in $\mathcal{P}$, the probability of
recovering $x_0$ using $\ell_1$ is given by
\[
P_{\ell_1}(A,s) = \frac{\card{\mathcal{F}_{s-1}(A\Cscr)}}
                       {\card{\mathcal{F}_{s-1}(\Cscr)}}.
\]
When we need to find the recoverability of vectors restricted to a
support $\Iscr$, this probability becomes
\begin{equation}\label{Eq:Pl1AI}
P_{\ell_1}(A,\mathcal{I})
 = \frac{\card{\mathcal{F}_{\mathcal{I}}(A\Cscr)}}
        {\card{\mathcal{F}_{\mathcal{I}}(\Cscr)}},
\end{equation}
where $\mathcal{F}_{\Iscr}(\Cscr) =
2^{\card{\Iscr}}$ denotes the number of faces in $\Cscr$
formed by the convex hull of $\{\pm e_j\}_{i\in\mathcal{I}}$, and
$\mathcal{F}_{\mathcal{I}}(A\Cscr)$ is the number of faces on
$A\Cscr$ generated by $\{\pm A\col{j}\}_{j\in\mathcal{I}}$.

\subsection{Null-space properties and $\ell_1$ recovery}

Equivalence results in terms of null-space properties generally
characterize equivalence for the set of all vectors $x$ with a fixed
support, which is defined as
\[
\support(x) = \{j \mid x_j \neq 0\}.
\]
We say that $x$ can be uniformly recovered on $\mathcal{I}\subseteq
\{1,\ldots,n\}$ if all $x$ with $\support(x) \subseteq\mathcal{I}$ can
be recovered. The following theorem illustrates conditions for uniform
recovery via $\ell_1$ on an index set; more general results are given
by Gribonval and Nielsen \cite{GRI2007Na}.

\begin{theorem}[Donoho and Elad \cite{DON2003Ea}, Gribonval and Nielsen
    \cite{GRI2003Na}]\label{Thm:NullSpaceL1}
    Let $A$ be an $m\times n$ matrix and $\mathcal{I} \subseteq
    \{1,\ldots,n\}$ be a fixed index set. Then all $x_0\in\Real^n$
    with $\support(x_0) \subseteq \mathcal{I}$ can be uniquely
    recovered from $b = Ax_0$ using basis pursuit \eqref{Eq:L1} if and
    only if for all $z \in \kernel(A)\setminus \{0\}$,
    \begin{equation}\label{Eq:NullSpaceCondL1}
      \sum_{j\in\mathcal{I}} \abs{z_j} < \sum_{j\not\in\mathcal{I}} \abs{z_j}.
    \end{equation}
    That is, the $\ell_1$-norm of $z$ on $\mathcal{I}$ is strictly less
    than the $\ell_1$-norm of $z$ on the complement $\mathcal{I}^c$.
\end{theorem}

\subsection{Optimality conditions for $\ell_1$ recovery}\label{Sec:OptimalityCondL1}

Sufficient conditions for recovery can be derived from the first-order
optimality conditions necessary for $x^*$ and $y^*$ to be solutions of
\eqref{Eq:L1} and \eqref{Eq:L1Dual} respectively. The
Karush-Kuhn-Tucker (KKT) conditions are also sufficient in this case
because the problems are convex. The Lagrangian function for
\eqref{Eq:L1} is given by
\[
\mathcal{L}(x,y) = \norm{x}_1 - y\T (Ax - b);
\]
the KKT conditions require that
\begin{equation}\label{Eq:KKTNSC}
Ax = b \text{and} 0\in \partial_x\mathcal{L}(x,y),
\end{equation}
where $\partial_x\mathcal{L}$ denotes the subdifferential of
$\mathcal{L}$ with respect to $x$. The second condition reduces to
\[
0 \in \sgn(x) - A\T y,
\]
where the signum function
\[
\sgn(\gamma) \in
\begin{cases}
   \sign(\gamma) & \hbox{if $\gamma \neq 0$,}
\\ [-1,1]        & \hbox{otherwise},
\end{cases}
\]
is applied to each individual component of $x$. It follows that $x^*$
is a solution of~\eqref{Eq:L1} if and only if $A\xstar=b$ and there
exists an $m$-vector $y$ such that $\abs{a_j\T y} \leq 1$ for
$j\not\in\support(x)$, and $a_j\T y = \sign(x_j^*)$ for all
$j\in\support(x)$. Fuchs~\cite{FUC2004a} shows that $x^*$ is the
unique solution of \eqref{Eq:L1} when $[a_j]_{j\in\support(x)}$ is
full rank and, in addition, $\abs{a_j\T y} < 1$ for all
$j\not\in\support(x)$. When the columns of $A$ are in general position
(i.e., no $k+1$ columns of $A$ span the same $k-1$ dimensional
hyperplane for $k \leq n$) we can weaken this condition by noting that
for such $A$, the solution of \eqref{Eq:L1} is always unique, thus
making the existence of a $y$ that satisfies~\eqref{Eq:KKTNSC} for $x_0$
a necessary and sufficient condition for $\ell_1$ to recover $x_0$.

\section{Recovery using sums-of-row norms}\label{Sec:SumOfRowNorms}

Our analysis of sparse recovery for the MMV problem of recovering
$X_0$ from $B=AX_0$ begins with an extension of
Theorem~\ref{Thm:NullSpaceL1} to recovery using the convex relaxation
\begin{equation}\label{Eq:SumOfNorms}
\minimize{X}\quad \sum_{j=1}^n\norm{X\row{j}}\quad\st\quad AX=B;
\end{equation}
note that the norm within the summation is arbitrary. Define
the row support of a matrix as
\[
\support_{\mathrm{row}}(X) = \{j \mid \norm{X\row{j}} \neq 0\}.
\]
With these definitions we have the following result. (A related result
is given by Stojnic et al.~\cite{STO2008PHa}.)

\begin{theorem}\label{Thm:NullSpaceMMV}
  Let $A$ be an $m\times n$ matrix, $k$ be a positive integer,
  $\mathcal{I} \subseteq \{1,\ldots,n\}$ be a fixed index set, and let
  $\norm{\cdot}$ denote any vector norm.  Then all $X_0 \in
  \Real^{n\times r}$ with $\support_{\mathrm{row}}(X_0)
  \subseteq\Iscr$ can be uniquely recovered from $B = AX_0$
  using~\eqref{Eq:SumOfNorms} if and only if for all $Z$ with columns
  $Z\col{k} \in \kernel(A)\setminus\{0\}$,
  \begin{equation}\label{Eq:NullSpaceCond}
    \sum_{j\in\mathcal{I}}\norm{Z\row{j}}
    < \sum_{j\not\in\mathcal{I}}\norm{Z\row{j}}.
  \end{equation}
\end{theorem}

\begin{proof}
For the ``only if'' part, suppose that there is a $Z$ with columns
$Z\col{k}\in\kernel(A)\setminus\{0\}$ such
that~\eqref{Eq:NullSpaceCond} does not hold. Now, choose $X\row{j} =
Z\row{j}$ for all $j\in\mathcal{I}$ and with all remaining rows
zero. Set $B = AX$. Next, define $V = X - Z$, and note that $AV = AX
- AZ = AX = B$. The construction of $V$ implies that
$\sum_{j}\norm{X\row{j}} \geq \sum_{j}\norm{V\row{j}}$, and
consequently $X$ cannot be the unique solution of~\eqref{Eq:SumOfNorms}.

Conversely, let $X$ be an arbitrary matrix with
$\support_{\mathrm{row}}(X)\subseteq \mathcal{I}$, and let $B =
AX$. To show that $X$ is the unique solution of~\eqref{Eq:SumOfNorms}
it suffices to show that for any $Z$ with columns
$Z\col{k}\in\kernel(A)\setminus\{0\}$,
\[
\sum_{j}\norm{(X+Z)\row{j}} > \sum_{j}\norm{X\row{j}}.
\]
This is equivalent to
\[
\sum_{j\not\in\mathcal{I}}\norm{Z\row{j}} +
\sum_{j\in\mathcal{I}}\norm{(X+Z)\row{j}} -
\sum_{j\in\mathcal{I}}\norm{X\row{j}} > 0.
\]
Applying the reverse triangle inequality, $\norm{a+b} - \norm{b} \geq
-\norm{a}$, to the summation over $j\in\mathcal{I}$ and reordering
exactly gives condition \eqref{Eq:NullSpaceCond}.
\end{proof}

In the special case of the sum of $\ell_1$-norms, i.e., $\ell_{1,1}$,
summing the norms of the columns is equivalent to summing the norms of
the rows. As a result, \eqref{Eq:SumOfNorms} can be written as
\[
  \minimize{X}\quad \sum_{k=1}^r\norm{X\col{k}}_1
  \quad\st\quad
  AX\col{k} = B\col{k},\quad k=1,\ldots,r.
\]
Because this objective is separable, the problem can be decoupled and
solved as a series of independent basis pursuit problems, giving one
$X\col{k}$ for each column $B\col{k}$ of $B$. The following result
relates recovery using the sum-of-norms formulation
\eqref{Eq:SumOfNorms} to $\ell_{1,1}$ recovery.

\begin{theorem}\label{Thm:SumNormsImplication}
  Let $A$ be an $m\times n$ matrix, $r$ be a positive integer,
  $\mathcal{I} \subseteq \{1,\ldots,n\}$ be a fixed index set, and
  $\norm{\cdot}$ denote any vector norm. Then uniform recovery of all
  $X \in \Real^{n\times r}$ with $\support_{\mathrm{row}}(X)
  \subseteq \mathcal{I}$ using sums of norms \eqref{Eq:SumOfNorms}
  implies uniform recovery on $\mathcal{I}$ using $\ell_{1,1}$.
\end{theorem}
\begin{proof}
For uniform recovery on support $\mathcal{I}$ to hold it follows
from Theorem \ref{Thm:NullSpaceMMV} that for any matrix $Z$ with
columns $Z\col{k}\in\kernel(A)\setminus\{0\}$,
property~\eqref{Eq:NullSpaceCond} holds. In particular it holds for
$Z$ with $Z\col{k} = \zbar$ for all $k$, with
$\zbar\in\kernel(A)\setminus\{0\}$. Note that for these matrices there
exist a norm-dependent constant $\gamma$ such that
\[
\abs{\zbar_j} = \gamma \norm{Z\row{j}}.
\]
Since the choice of $\zbar$ was arbitrary, it follows
from~\eqref{Eq:NullSpaceCond} that the
NS-condition~\eqref{Eq:NullSpaceCondL1} for independent recovery of
vectors $B\col{k}$ using $\ell_1$ in Theorem \ref{Thm:NullSpaceL1} is
satisfied. Moreover, because $\ell_{1,1}$ is equivalent to independent
recovery, we also have uniform recovery on $\mathcal{I}$ using
$\ell_{1,1}$.
\end{proof}

An implication of Theorem~\ref{Thm:SumNormsImplication} is that the
use of restricted isometry conditions---or any technique, for that
matter---to analyze uniform recovery conditions for the sum-of-norms
approach necessarily lead to results that are no stronger than uniform
$\ell_1$ recovery. (Recall that the $\ell_{1,1}$ and $\ell_1$ norms
are equivalent).

\subsection{Recovery using $\ell_{1,2}$}

\begin{figure}[t]
\centering
\includegraphics[width=0.6\textwidth]{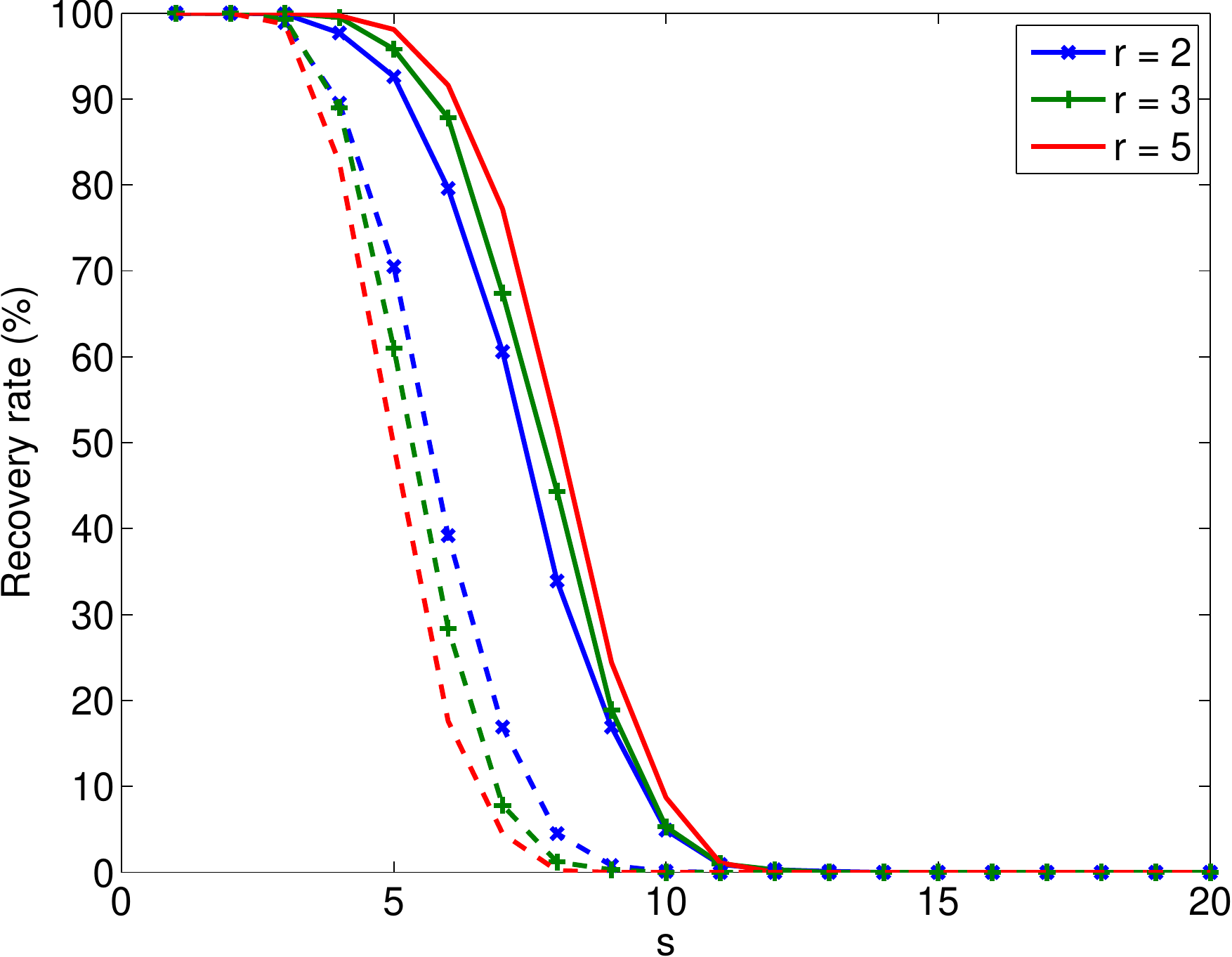}
\caption{Recovery rates for fixed, randomly drawn $20\times 60$
  matrices $A$, averaged over 1,000 trials at each row-sparsity level
  $s$. The nonzero entries in the $60\times r$ matrix $X_0$ are
  sampled i.i.d.\@ from the normal
  distribution. The solid and dashed lines represent $\ell_{1,2}$ and
  $\ell_{1,1}$ recovery, respectively.}\label{Fig:L11versusL12}
\end{figure}

In this section we take a closer look at the $\ell_{1,2}$ problem
\begin{equation}\label{Eq:L12}
\minimize{X}\quad \norm{X}_{1,2}\quad \st\quad AX=B,
\end{equation}
which is a special case of the sum-of-norms problem. Although
Theorem~\ref{Thm:SumNormsImplication} establishes that uniform
recovery via $\ell_{1,2}$ is no better than uniform recovery via
$\ell_{1,1}$, there are many situations in which it recovers signals
that $\ell_{1,1}$ cannot. Indeed, it is evident from
Figure~\ref{Fig:L11versusL12} that the probability of recovering
individual signals with random signs and support is much higher for
$\ell_{1,2}$. The reason for the degrading performance or $\ell_{1,1}$
with increasing $k$ is explained in Section~\ref{Sec:BoostedL1}.

In this section we construct examples for which $\ell_{1,2}$ works and
$\ell_{1,1}$ fails, and vice versa. This helps uncover some of the
structure of $\ell_{1,2}$, but at the same time implies that certain
techniques used to study $\ell_1$ can no longer be used
directly. Because the examples are based on extensions of the results
from Section \ref{Sec:OptimalityCondL1}, we first develop equivalent
conditions here.

\subsubsection{Sufficient conditions for recovery via $\ell_{1,2}$}
\label{Sec:OptimalityL12}

The optimality conditions of the $\ell_{1,2}$ problem~\eqref{Eq:L12}
play a vital role in deriving a set of sufficient conditions for
joint-sparse recovery.  In this section we derive the dual
of~\eqref{Eq:L12} and the corresponding necessary and sufficient
optimality conditions.  These allow us to derive sufficient conditions
for recovery via $\ell_{1,2}$.

The Lagrangian for~\eqref{Eq:L12} is defined as
\begin{equation}\label{Eq:LagrangeL12}
\Lscr(X,Y) = \norm{X}_{1,2} - \iprod{Y}{AX-B},
\end{equation}
where $\iprod{V}{W}\defd\trace(V\T W)$ is an inner-product defined over
real matrices.  The dual is then given by maximizing
\begin{equation} \label{Eq:ytbsup}
  \inf_{X}\Lscr(X,Y)
   = \inf_{X}\left\{\norm{X}_{1,2}-\iprod{Y}{AX-B}\right\}
   = \iprod{B}{Y} - \sup_{X} \left\{\iprod{A\T Y}{X} -
       \norm{X}_{1,2}\right\} 
\end{equation}
over $Y$. (Because the primal problem has only linear
constraints, there necessarily exists a dual solution $Y^*$ that
maximizes this expression \cite[Theorem 28.2]{ROC70}.)  To simplify
the supremum term, we note that for any convex, positively homogeneous
function $f$ defined over an inner-product space,
$$
\sup_v\ \{\iprod w v - f(v)\} =
\begin{cases}
  0      & \hbox{if $w\in\partial f(0)$,}
\\\infty & \hbox{otherwise.}
\end{cases}
$$
To derive these conditions, note that positive homogeneity of $f$
implies that $f(0)=0$, and thus $w\in\partial f(0)$ implies that
$\iprod w v \le f(v)$ for all $v$.  Hence, the supremum is achieved
with $v=0$.  If on the other hand $w\not\in\partial f(0)$, then there
exists some $v$ such that $\iprod w v > f(v)$, and by the positive
homogeneity of $f$, $\iprod{w}{\alpha v}-f(\alpha v)\to\infty$ as
$\alpha\to\infty$.  Applying this expression for the supremum
to~\eqref{Eq:ytbsup}, we arrive at the necessary condition
\begin{equation} \label{eq:2}
  A\T Y\in\partial\norm{0}_{1,2},
\end{equation}
which is required for dual feasibility.

We now derive an expression for the subdifferential
$\partial\norm{X}_{1,2}$. For rows $j$ where $\norm{X\row j}_2>0$, the
gradient is given by $\nabla\norm{X\row j}_2=X\row j/\norm{X\row
  j}_2$.  For the remaining rows, the gradient is not defined, but
$\partial\norm{X\row j}_2$ coincides with the set of unit
$\ell_2$-norm vectors $\Bscr_{\ell_2}^r = \{v\in\Real^r\ \mid
\norm{v}_2 \leq 1\}$. Thus, for each $j=1,\ldots,n$,
\begin{equation} \label{eq:1}
\partial_{X\row{j}}\norm{X}_{1,2} \in
\begin{cases}
  X\row{j}/\norm{X\row{j}}_2 & \hbox{if $\norm{X\row{j}}_2 > 0$,}
\\[4pt] \Bscr_{\ell_2}^r     & \hbox{otherwise.}
\end{cases}
\end{equation}
Combining this expression with \eqref{eq:2}, we arrive at the dual
of~\eqref{Eq:L12}:
\begin{equation} \label{Eq:L12MatrixDual}
  \maximize{Y}\quad\trace(B\T Y)
  \quad\st\quad\norm{A\T Y}_{\infty,2}\leq 1.
\end{equation}
The following conditions are therefore necessary and sufficient for a
primal-dual pair $(X^*,Y^*)$ to be optimal for~\eqref{Eq:L12} and its
dual~\eqref{Eq:L12MatrixDual}:
\begin{subequations}                       \label{Eq:NSCL12}
\begin{alignat}{2}
  AX^* &= B
  &&\qquad\hbox{(primal feasibility);} \label{Eq:NSCL12a}
\\\norm{A\T Y^*}_{\infty,2} &\le 1
  &&\qquad\hbox{(dual feasibility);}   \label{Eq:NSCL12b}
\\\norm{X^*}_{1,2} &= \trace(B\T Y^*)
  &&\qquad\hbox{(zero duality gap).}   \label{Eq:NSCL12c}
\end{alignat}
\end{subequations}

The existence of a matrix $Y^*$ that satisfies~\eqref{Eq:NSCL12}
provides a certificate that the feasible matrix $X^*$ is an optimal
solution of~\eqref{Eq:L12}. However, it does not guarantee that $X^*$
is also the unique solution. The following theorem gives sufficient
conditions, similar to those in Section \ref{Sec:OptimalityCondL1},
that also guarantee uniqueness of the solution.

\begin{theorem}
  Let $A$ be an $m\times n$ matrix, and $B$ be an $m\times r$ matrix.
  Then a set of sufficient conditions for $X$ to be the unique
  minimizer of~\eqref{Eq:L12} with Lagrange multiplier
  $Y\in\Real^{m\times r}$ and row support
  $\Iscr=\support_{\mathrm{row}}(X)$, is that
\begin{subequations}                    \label{Eq:SuffL12}
\begin{alignat}{2}
   & AX=B,                              \label{Eq:SuffL12a}
\\ & (A\T Y)\col{j} = (X^*)\row{j}/\norm{(X^*)\row{j}}_2,
   &\qquad j &\in\Iscr                  \label{Eq:SuffL12b}
\\ &\norm{(A\T Y)\col{j}}_2 < 1,
   &\qquad j &\not\in\Iscr              \label{Eq:SuffL12c}
\\ &\rank(A_\Iscr) = \vert\Iscr\vert.   \label{Eq:SuffL12d}
\end{alignat}
\end{subequations}
\end{theorem}
\begin{proof}
  The first three conditions clearly imply that $(X,Y)$ primal and
  dual feasible, and thus satisfy~\eqref{Eq:NSCL12a}
  and~\eqref{Eq:NSCL12b}.  Conditions~\eqref{Eq:SuffL12b}
  and~\eqref{Eq:SuffL12c} together imply that
  $$
  \trace(B\T Y) \equiv \sum_{j=1}^n[(A\T Y)\col{j}]^T X\row j =
  \sum_{j=1}^n X\row j \equiv \norm{X}_{1,2}.
  $$
  The first and last identities above follow directly from the
  definitions of the matrix trace and of the norm
  $\norm{\cdot}_{1,2}$, respectively; the middle equality follows from
  the standard Cauchy inequality. Thus, the zero-gap requirement
  \eqref{Eq:NSCL12c} is satisfied. The
  conditions~\eqref{Eq:SuffL12a}--\eqref{Eq:SuffL12c} are therefore
  sufficient for $(X,Y)$ to be an optimal primal-dual solution
  of~\eqref{Eq:L12}.  Because $Y$ determines the support and is a
  Lagrange multiplier for every solution $X$, this support must be
  unique. It then follows from condition~\eqref{Eq:SuffL12d} that $X$
  must be unique.
\end{proof}

\subsection{Counter examples}

Using the sufficient and necessary conditions developed in the
previous section we now construct examples of problems for which
$\ell_{1,2}$ succeeds while $\ell_{1,1}$ fails, and vice
versa. Because of its simplicity, we begin with the latter.

\paragraph{Recovery using $\ell_{1,1}$ where $\ell_{1,2}$ fails.}
Let $A$ be an $m\times n$ matrix with $m < n$ and unit-norm columns
that are not scalar multiples of each other. Take any vector
$x\in\Real^n$ with at least $m+1$ nonzero entries. Then $X_0 =
\diag(x)$, possibly with all identically zero columns removed, can be
recovered from $B = AX_0$ using $\ell_{1,1}$, but not with
$\ell_{1,2}$. To see why, note that each column in $X_0$ has only a
single nonzero entry, and that, under the assumptions on $A$, each
one-sparse vector can be recovered individually using $\ell_1$ (the
points $\pm A\col{j} \in \Real^m$ are all $0$-faces of
$\mathcal{P}$) and therefore that $X_0$ can be recovered using
$\ell_{1,1}$.

On the other hand, for recovery using $\ell_{1,2}$ there would need to
exist a matrix $Y$ satisfying the first condition of \eqref{Eq:NSCL12}
for all $j \in \mathcal{I} = \{1,\ldots,n\}$. For this given $X_0$ this
reduces to $A^TY = M$, where $M$ is the identity matrix, with the same
columns removed as $X$. But this equality is impossible to satisfy
because $\rank(A) \leq m < m+1 \leq \rank(M)$. Thus, $X_0$ cannot be
the solution of the $\ell_{1,2}$ problem~\eqref{Eq:L12}.

\paragraph{Recovery using $\ell_{1,2}$ where $\ell_{1,1}$ fails.}

For the construction of a problem where $\ell_{1,2}$ succeeds and
$\ell_{1,1}$ fails, we consider two vectors, $f$ and $s$, with the
same support $\mathcal{I}$, in such a way that individual $\ell_1$
recovery fails for $f$, while it succeeds for $s$. In addition we
assume that there exists a vector $y$ that satisfies
$$
 y\T A\col{j} = \sign(s_j) \quad\hbox{for all $j\in\mathcal{I}$,}
 \textt{and}
 \abs{y\T A\col{j}} < 1 \quad\hbox{for all $j\not\in\mathcal{I}$;}
$$
i.e., $y$ satisfies conditions~\eqref{Eq:SuffL12b}
and~\eqref{Eq:SuffL12c}. Using the vectors $f$ and $s$, we construct
the 2-column matrix $X_0 = [(1-\gamma) s,\ \gamma f]$, and claim that
for sufficiently small $\gamma > 0$, this gives the desired
reconstruction problem. Clearly, for any $\gamma \neq 0$, $\ell_{1,1}$
recovery fails because the second column can never be recovered, and
we only need to show that $\ell_{1,2}$ does succeed.

For $\gamma = 0$, the matrix $Y = [y, 0]$ satisfies
conditions~\eqref{Eq:SuffL12b} and~\eqref{Eq:SuffL12c} and, assuming
\eqref{Eq:SuffL12d} is also satisfied, $X_0$ is the unique solution of
$\ell_{1,2}$ with $B = AX_0$. For sufficiently small $\gamma > 0$, the
conditions that $Y$ need to satisfy change slightly due to the
division by $\norm{X_0\row{j}}_2$ for those rows in
$\support_{\mathrm{row}}(X)$. By adding corrections to the columns of
$Y$ those new conditions can be satisfied. In particular, these
corrections can be done by adding weighted combinations of the columns
in $\bar{Y}$, which are constructed in such a way that it satisfies
$A_\mathcal{I}^T \Ybar = I$, and minimizes
$\norm{A_{\mathcal{I}^c}^T\bar{Y}}_{\infty,\infty}$ on the complement
$\mathcal{I}^c$ of $\mathcal{I}$.

Note that on the above argument can also be used to show that
$\ell_{1,2}$ fails for $\gamma$ sufficiently close to one. Because the
support and signs of $X$ remain the same for all $0 < \gamma < 1$, we
can conclude the following:
\begin{corollary}
  Recovery using $\ell_{1,2}$ is generally not only characterized
  by the row-support and the sign pattern of the nonzero entries in
  $X_0$, but also by the magnitude of the nonzero entries.
\end{corollary}

A consequence of this conclusion is that the notion of faces used in
the geometrical interpretation of $\ell_1$ is not applicable to the
$\ell_{1,2}$ problem.

\subsection{Experiments}\label{Sec:L12Experiments}

To get an idea of just how much more $\ell_{1,2}$ can recover in the
above case where $\ell_{1,1}$ fails, we generated a $20\times 60$
matrix $A$ with entries i.i.d.\@ normally distributed, and determined
a set of vectors $s_i$ and $f_i$ with identical support for which
$\ell_{1}$ recovery succeeds and fails, respectively. Using
triples of vectors $s_i$ and $f_j$ we constructed row-sparse matrices
such as $X_0 = [s_1,f_1,f_2]$ or $X_0 = [s_1,s_2,f_2]$, and attempted
to recover from $B = AX_0W$, where $W =
\diag(\omega_1,\omega_2,\omega_3)$ is a diagonal weighting matrix with
nonnegative entries and unit trace, by solving \eqref{Eq:L12}. For
problems of this size, interior-point methods are very efficient and
we use SDPT3~\cite{Software:SDPT3} through the CVX
interface~\cite{Software:CVX,GRA2008Ba}. We consider $X_0$ to be
recovered when the maximum absolute difference between $X_0$ and the
$\ell_{1,2}$ solution $X^*$ is less than $10^{-5}$. The results of the
experiment are shown in Figure \ref{Fig:Triangles}. In addition to the
expected regions of recovery around individual columns $s_i$ and
failure around $f_i$, we see that certain combinations of vectors
$s_i$ still fail, while other combinations of vectors $f_i$ may be
recoverable. By contrast, when using $\ell_{1,1}$ to solve the
problem, any combination of $s_i$ vectors can be recovered while no
combination including an $f_i$ can be recovered.

\begin{figure}[t]
\centering
\begin{tabular}{@{}cccc@{}}
\includegraphics[width=0.23\textwidth]{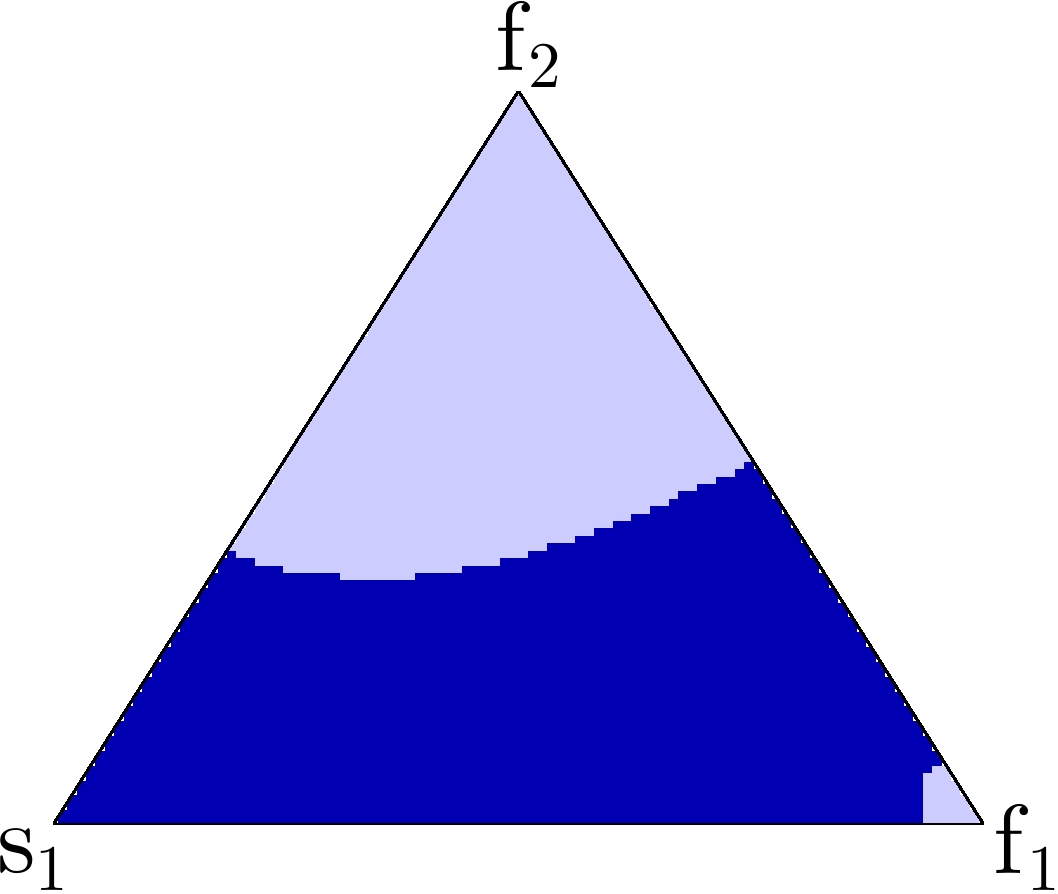} &
\includegraphics[width=0.23\textwidth]{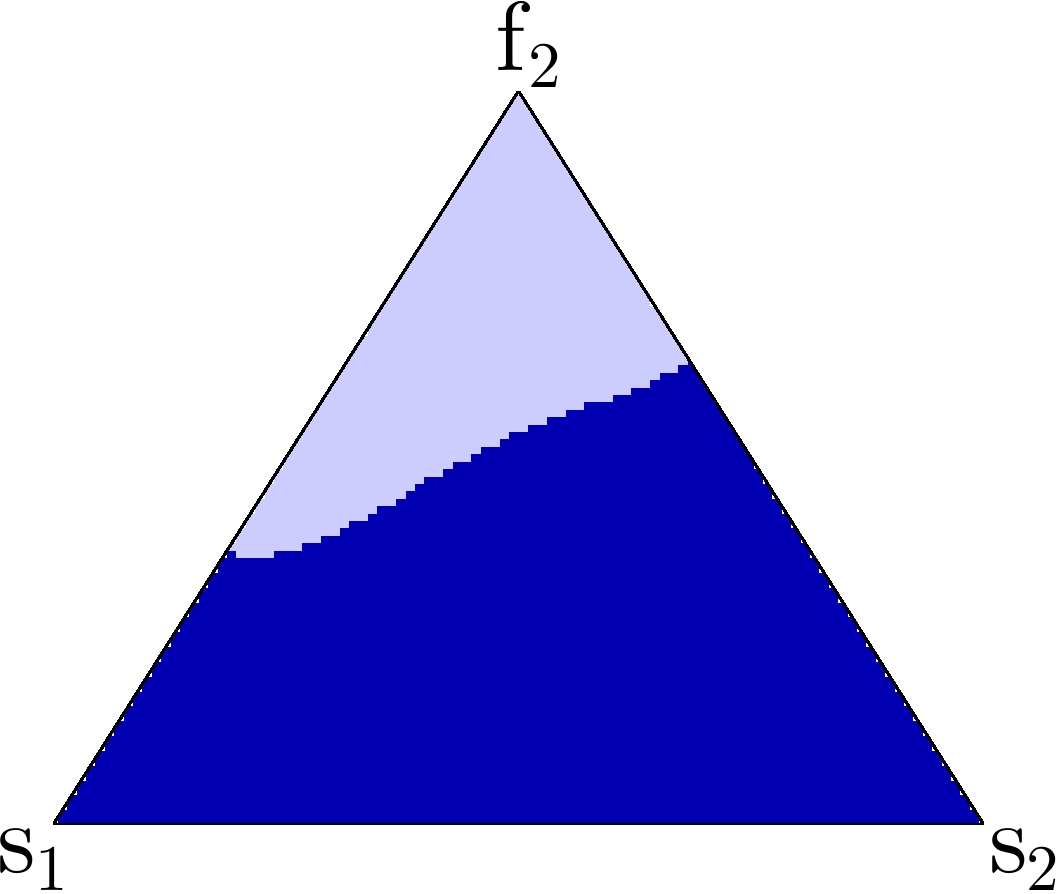} &
\includegraphics[width=0.23\textwidth]{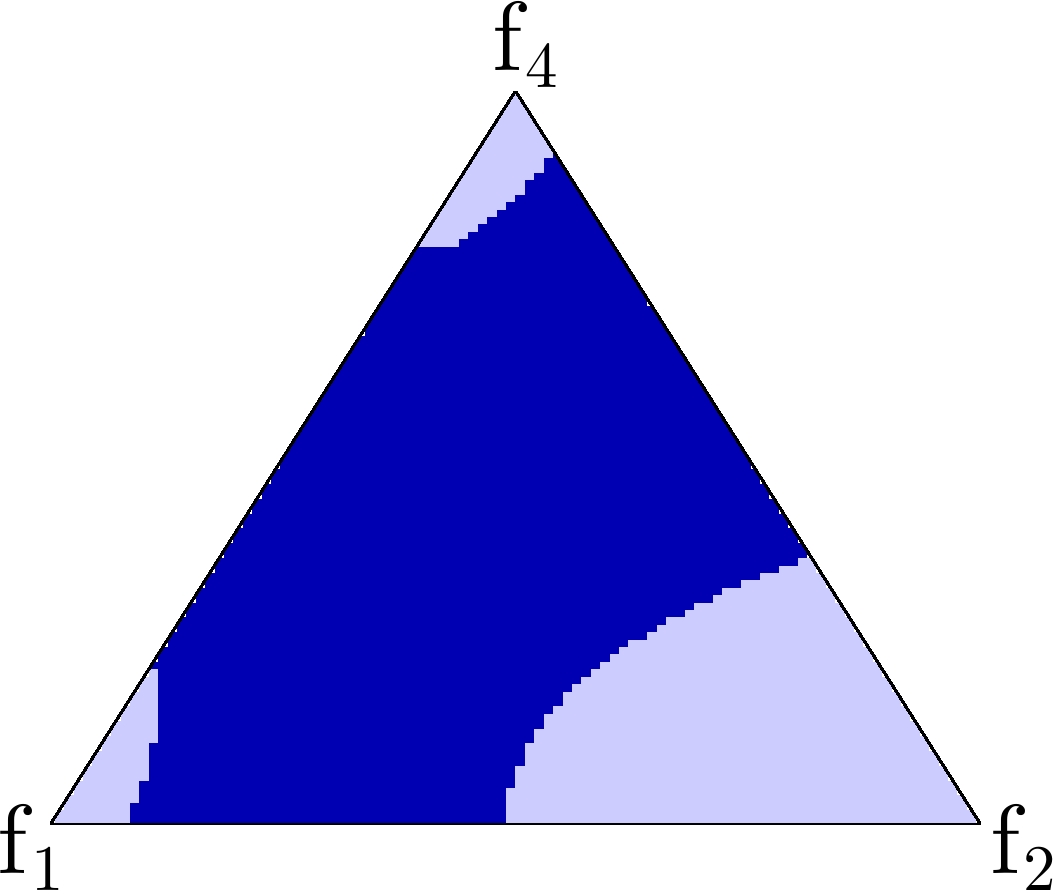}&
\includegraphics[width=0.23\textwidth]{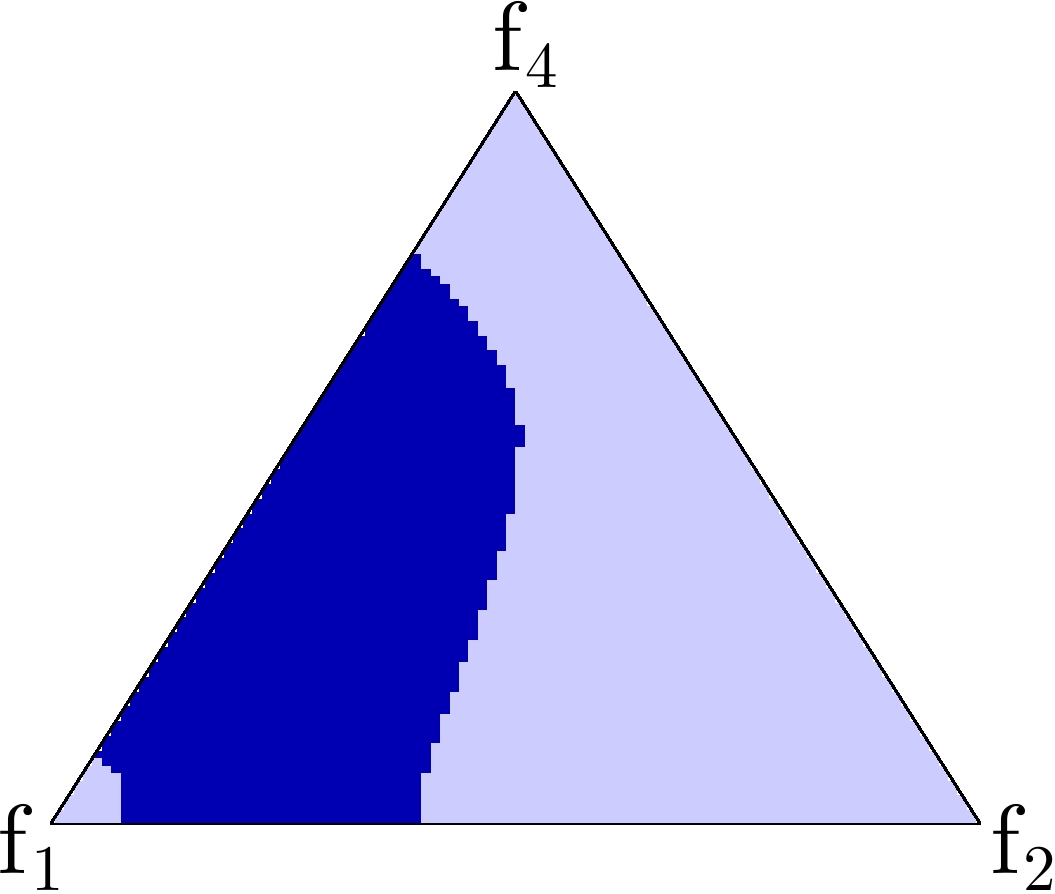}\\
$\abs{\Iscr}=5$ & $\abs{\Iscr}=5$ & $\abs{\Iscr}=5$ & $\abs{\Iscr}=7$
\\[6pt]
\includegraphics[width=0.23\textwidth]{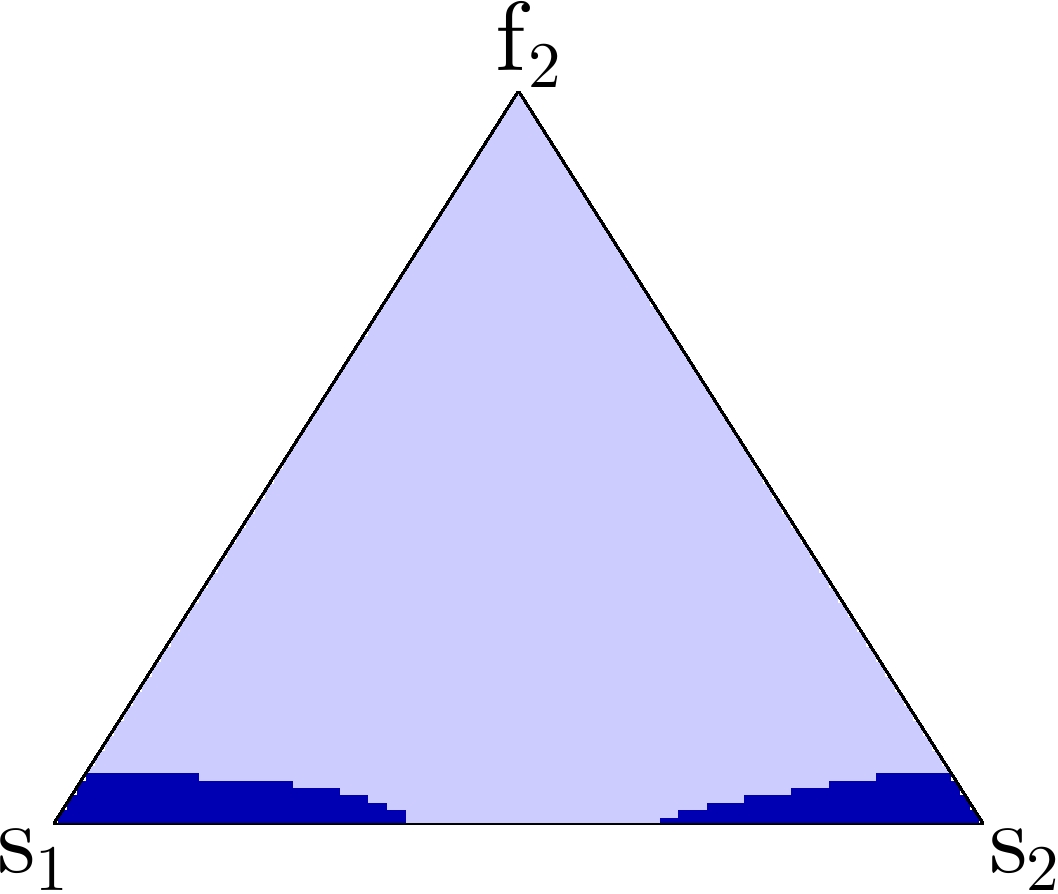}&
\includegraphics[width=0.23\textwidth]{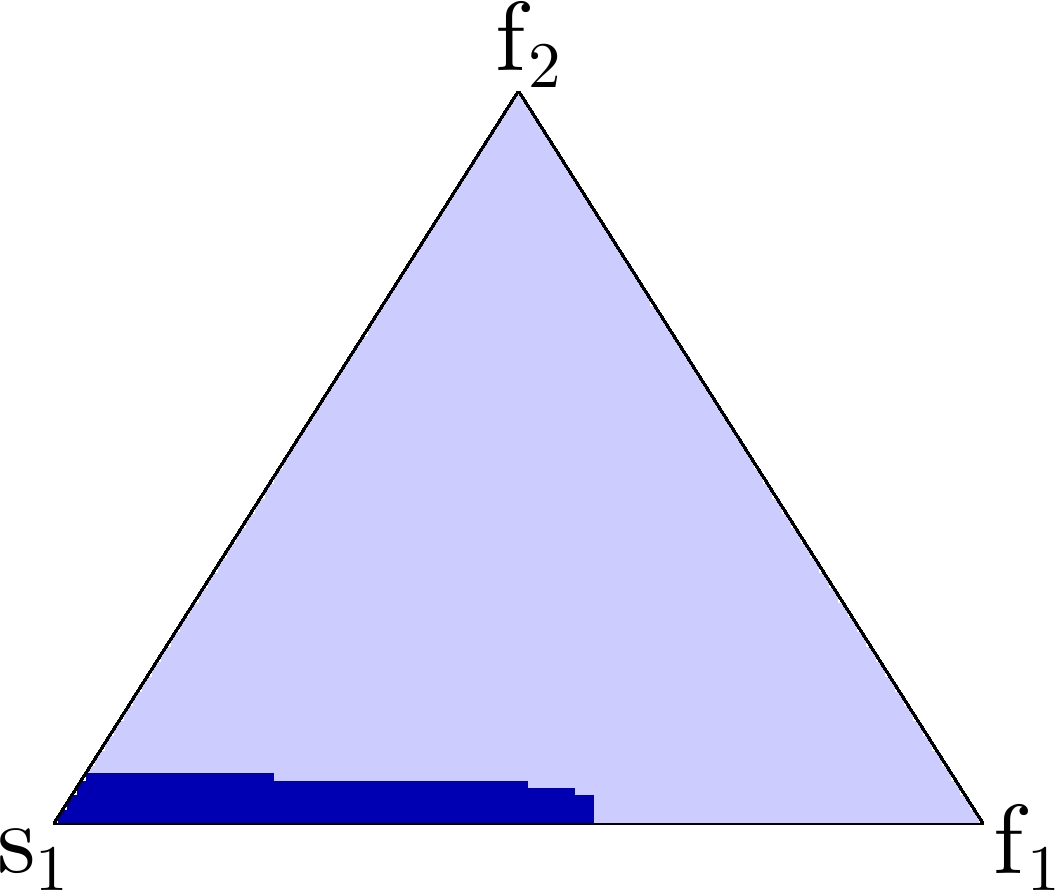}& 
\includegraphics[width=0.23\textwidth]{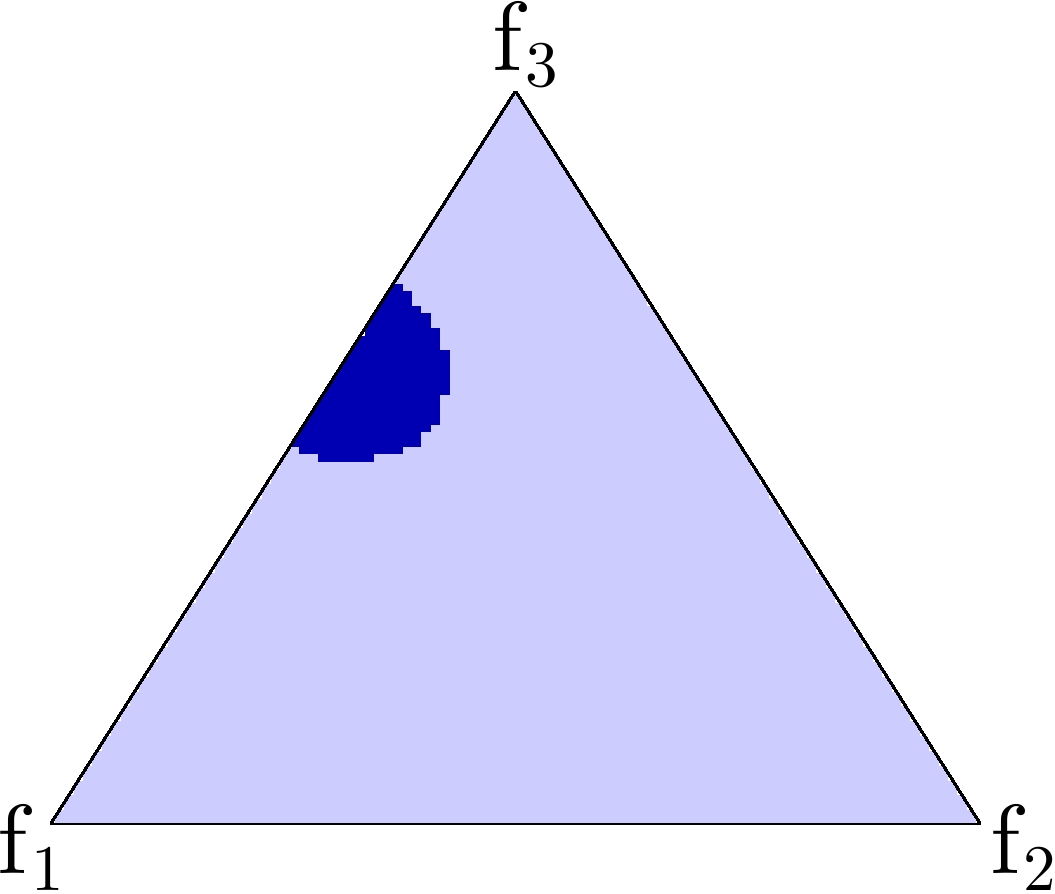}&
\includegraphics[width=0.23\textwidth]{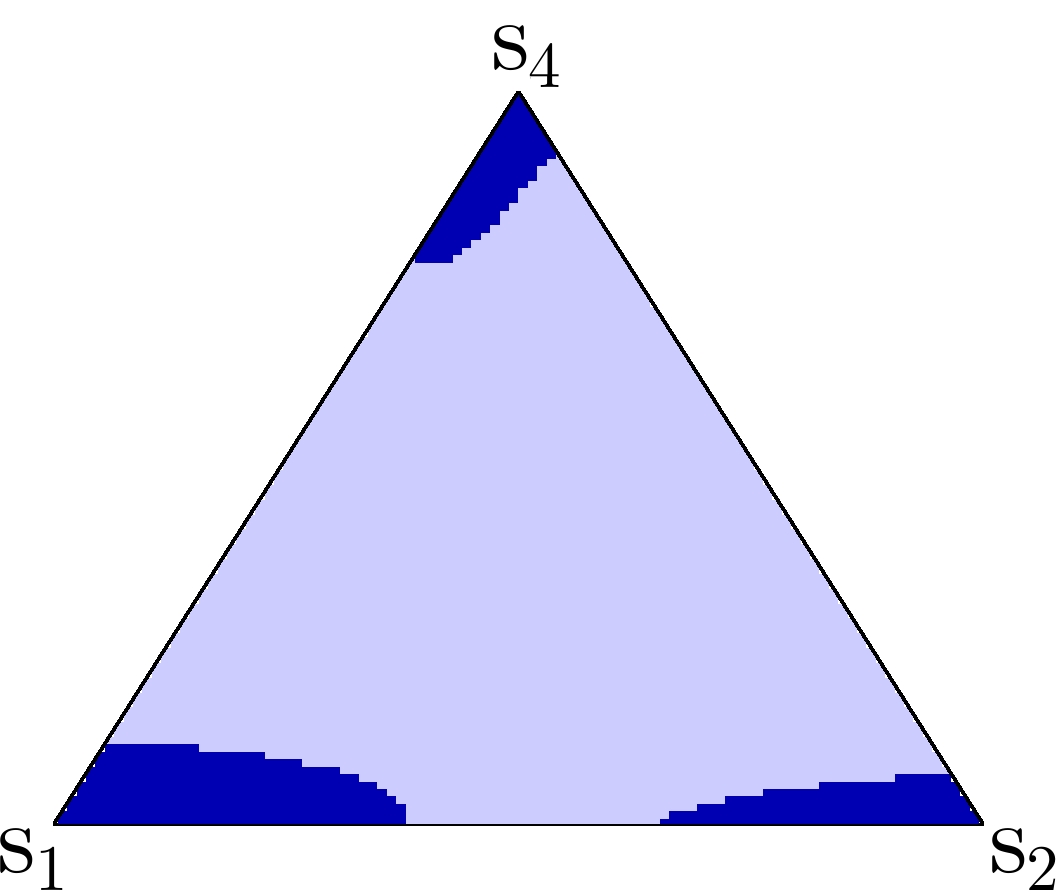}\\
$\abs{\Iscr}=10$ & $\abs{\Iscr}=10$ & $\abs{\Iscr}=10$ & $\abs{\Iscr}=10$
\end{tabular}
\caption{Generation of problems where $\ell_{1,2}$ succeeds, while
  $\ell_{1,1}$ fails. For a $20\times 60$ matrix $A$ and fixed
  support of size $\abs{\Iscr}=5,7,10$, we create vectors $f_i$ that
  cannot be recovered using $\ell_1$, and vectors $s_i$ than can be
  recovered. Each triangle represents an $X_0$ constructed from the
  vectors denoted in the corners. The location in the triangle
  determines the weight on each vector, ranging from zero to one, and
  summing up to one. The dark areas indicates the weights for which
  $\ell_{1,2}$ successfully recovered $X_0$.}
\label{Fig:Triangles}
\end{figure}

\section{Boosted $\ell_1$}\label{Sec:BoostedL1}

As described in Section~\ref{Sec:SumOfRowNorms}, recovery using
$\ell_{1,1}$ is equivalent to individual $\ell_1$ recovery of each
column $x_k := X_0\col{k}$ based on $b_k\defd B\col{k}$, for
$k=1,\ldots,r$:
\begin{equation}\label{Eq:L1i}
\minimize{x}\quad \norm{x}_1\quad\st\quad Ax = b_k.
\end{equation}
Assuming that the signs of nonzero entries in the support of each
$x_k$ are drawn i.i.d.\@ from $\{1,-1\}$, we can express the
probability of recovering a matrix $X_0$ with row support
$\mathcal{I}$ using $\ell_{1,1}$ in terms of the probability of
recovering vectors on that support using $\ell_1$. To see how, note
that $\ell_{1,1}$ recovers the original $X_0$ if and only if each
individual problem in \eqref{Eq:L1i} successfully recovers each
$x_k$. For the above class of matrices $X_0$ this therefore gives a
recovery rate of
\[
P_{\ell_{1,1}}(A,\mathcal{I},k) = \left[P_{\ell_1}(A,\mathcal{I})\right]^r.
\]
Using $\ell_{1,1}$ to recover $X_0$ is clearly not a good idea. Note
also that uniform recovery of $X_0$ on a support $\mathcal{I}$ remains
unchanged, regardless of the number of observations, $r$, that are
given. As a consequence of Theorem~\ref{Thm:SumNormsImplication}, this
also means that the uniform-recovery properties for any sum-of-norms
approach cannot increase with $r$. This clearly defeats the purpose of
gathering multiple observations.

In many instances where $\ell_{1,1}$ fails, it may still recover a
subset of columns $x_k$ from the corresponding observations $b_k$. It
seems wasteful to discard this information because if we could
recognize a single correctly recovered $x_k$, we would immediately
know the row support $\Iscr = \support_{\mathrm{row}}(X_0) =
\support(x_k)$ of $X_0$. Given the correct support we can recover the
nonzero part $\bar{X}$ of $X_0$ by solving
\begin{equation}\label{Eq:AIXB}
\minimize{\bar{X}}\quad \norm{A_{\mathcal{I}}\bar{X} - B}_{F}.
\end{equation}
In practice we obviously do not know the correct support, but when a
given solution $x_k^*$ of~\eqref{Eq:L1i} that is sufficiently sparse,
we can try to solve~\eqref{Eq:AIXB} for that support and verify if the
residual at the solution is zero. If so, we construct the final $X^*$
using the non-zero part and declare success. Otherwise we simply
increment $k$ and repeat this process until there are no more
observations and recovery was unsuccessful. We refer to this
algorithm, which is reminiscent of the \rembo{} approach
\cite{MIS2008Ea}, as boosted $\ell_1$; its sole aim is to provide a
bridge to the analysis of \rembo. The complete boosted $\ell_1$
algorithm is outlined in Figure~\ref{Alg:BoostedL1}.

The recovery properties of the boosted $\ell_1$ approach are opposite
from those of $\ell_{1,1}$: it fails only if all individual columns
fail to be recovered using $\ell_1$. Hence, given an unknown $n\times r$
matrix $X$ supported on $\mathcal{I}$ with its sign pattern uniformly
random, the boosted $\ell_1$ algorithm gives an expected recovery rate
of
\begin{equation}\label{Eq:PBoostedL1}
P_{\ell_1^B}(A,\mathcal{I},r) = 1-\left[1-P_{\ell_1}(A,\mathcal{I})\right]^r.
\end{equation}

To experimentally verify this recovery rate, we generated a $20\times
80$ matrix $A$ with entries independently sampled from the normal
distribution and fixed a randomly chosen support set $\mathcal{I}_s$
for three levels of sparsity, $s=8,9,10$. On each of these three
supports we generated vectors with all possible sign patterns and
solved~\eqref{Eq:L1} to see if they could be recovered or not (see
Section~\ref{Sec:L12Experiments}). This gives exactly the face counts
required to compute the $\ell_1$ recovery probability
in~\eqref{Eq:Pl1AI}, and the expected boosted $\ell_1$ recovery rate
in~\eqref{Eq:PBoostedL1}

For the empirical success rate we take the average over 1,000 trials
with random coefficient matrices $X$ supported on $\mathcal{I}_s$, and
its nonzero entries independently drawn from the normal
distribution. To reduce the computational time we avoid solving
$\ell_1$ and instead compare the sign pattern of the current solution
$x_k$ against the information computed to determine the face counts
(both $A$ and $\mathcal{I}_s$ remain fixed).  The theoretical and
empirical recovery rates using boosted $\ell_1$ are plotted in
Figure~\ref{Fig:BoostedL1}.

\begin{figure}
\begin{tabular}{ll}
\begin{minipage}{0.44\textwidth}
\fbox{\begin{minipage}{0.95\textwidth}
  \smallskip
  given $A$, $B$\\
  \For{$k=1,\ldots,r$}{
        solve \eqref{Eq:L1} with $b_k = B\col{k}$ to get $x$\\
        $\mathcal{I} \gets \support(x)$\\
	\If{$\vert\mathcal{I}\vert < m/2$}{
           solve \eqref{Eq:AIXB} to get $X$ \\
        \If{$A_{\mathcal{I}}X = B$}{
           $X^* = 0$\\
           $(X^*)\row{j} \gets X\row{j}$ for $j\in\Iscr$\\
           return solution $X^*$
        }}
  }
  return failure
  \smallskip
\end{minipage}}
\end{minipage} &
\begin{minipage}{0.53\textwidth}
\includegraphics[width=\textwidth]{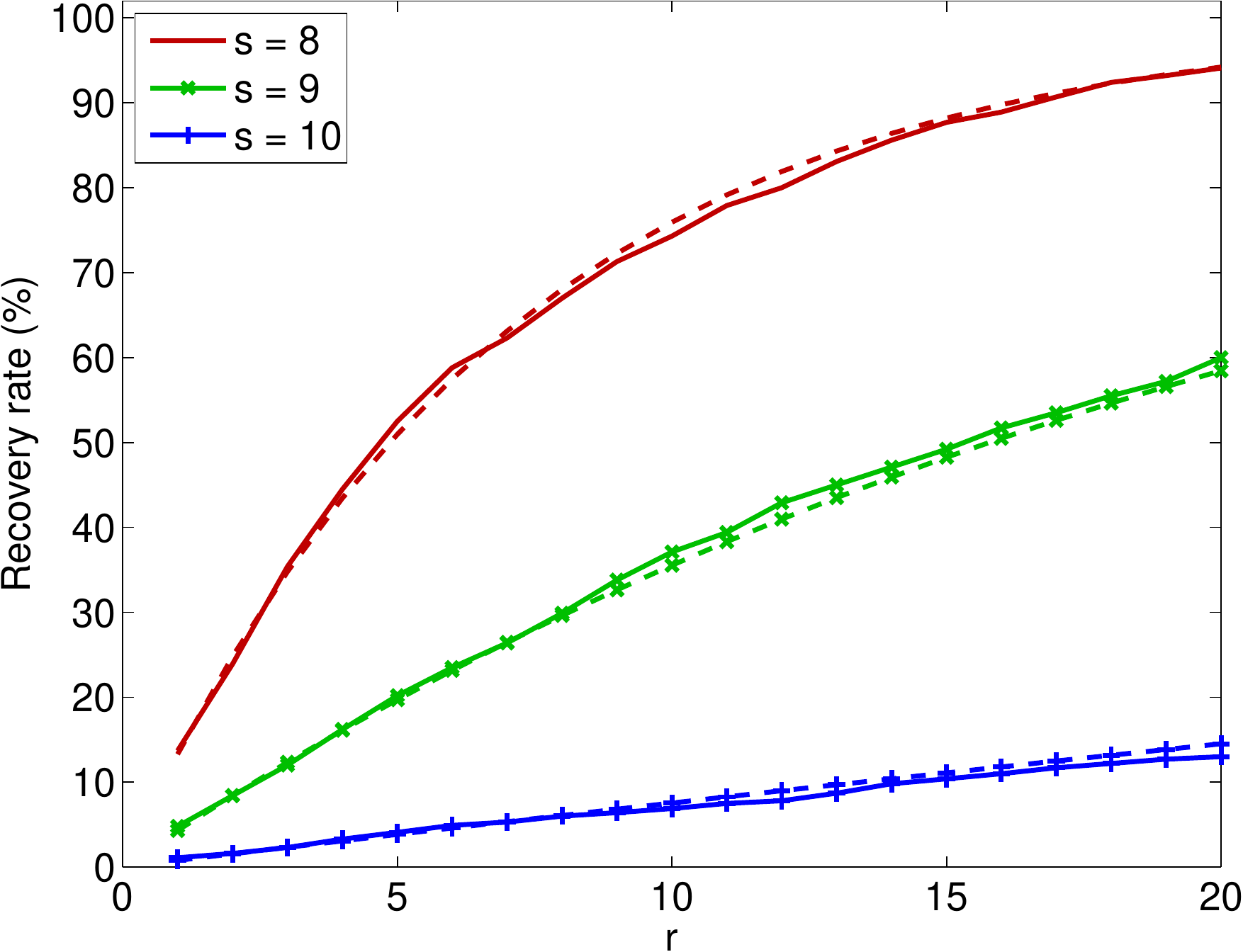}
\end{minipage} \\
\begin{minipage}[t]{0.44\textwidth}
\caption{The boosted $\ell_1$ algorithm}\label{Alg:BoostedL1}
\end{minipage}
 &
\begin{minipage}[t]{0.53\textwidth}
\caption{Theoretical (dashed) and experimental (solid) performance of
  boosted $\ell_1$ for three problem instances with different row
  support $s$.}\label{Fig:BoostedL1}
\end{minipage} \\
\end{tabular}
\end{figure}

\section{Recovery using \rembo}\label{Sec:ReMBo}

The boosted $\ell_1$ approach can be seen as a special case of the
\rembo~\cite{MIS2008Ea} algorithm. \rembo{} proceeds by taking a
random vector $w \in \Real^r$ and combining the individual
observations in $B$ into a single weighted observation $b \defd
Bw$. It then solves a single measurement vector problem $Ax=b$ for
this $b$ (we shall use $\ell_1$ throughout) and checks if the computed
solution $x^*$ is sufficiently sparse. If not, the above steps are
repeated with a different weight vector $w$; the algorithm stops when
a maximum number of trials is reached. If the support $\mathcal{I}$ of
$x^*$ is small, we form $A_{\mathcal{I}} = [A\col
j]_{j\in\mathcal{I}}$, and check if~\eqref{Eq:AIXB} has a solution
$\bar{X}$ with zero residual. If this is the case we have the nonzero
rows of the solution $X^*$ in $\bar{X}$ and are done. Otherwise, we
simply proceed with the next $w$. The \rembo{} algorithm reduces to
boosted $\ell_1$ by limiting the number of iterations to $r$ and
choosing $w = e_i$ in the $i$th iteration. We summarize the
\rembo-$\ell_1$ algorithm in Figure~\ref{Alg:ReMBo}. The formulation
given in~\cite{MIS2008Ea} requires a user-defined threshold on the
cardinality of the support $\mathcal{I}$ instead of the fixed
threshold $m/2$. Ideally this threshold should be half of the
spark~\cite{DON2003Ea} of A, where
\[
\spark(A) \defd \min_{z \in \kernel(A)\setminus\{0\}}\ \norm{z}_0
\]
which is the number of nonzeros of the sparsest vector in the kernel
of $A$; any vector $x_0$ with fewer than $\spark(A)/2$ nonzeros is the
unique sparsest solution of $Ax = Ax_0 =
b$~\cite{DON2003Ea}. Unfortunately, the spark is prohibitively
expensive to compute, but under the assumption that $A$ is in general
position, $\spark(A) = m+1$. Note that choosing a higher value can
help to recover signals with row sparsity exceeding $m/2$. However, in
this case it can no longer be guaranteed to be the sparsest solution.

\begin{figure}
\begin{tabular}{ll}
\begin{minipage}{0.44\textwidth}
\fbox{\begin{minipage}{0.95\textwidth}
  \smallskip
  given $A$, $B$. Set $\mathrm{Iteration} \gets 0$\\
  \While{$\mathrm{Iteration} < \mathrm{MaxIteration}$}{
	$w \gets \mathrm{Random}(n,1)$\\
        solve \eqref{Eq:L1} with $b = Bw$ to get $x$\\
        $\mathcal{I} \gets \support(x)$\\\
	\If{$\vert\mathcal{I}\vert < m/2$}{
           solve \eqref{Eq:AIXB} to get $X$ \\
        \If{$A_{\mathcal{I}}X = B$}{
           $X^* = 0$\\
           $(X^*)\row{j} \gets X\row{j}$ for $j\in\mathcal{I}$\\
           return solution $X^*$
        }}
	$\mathrm{Iteration} \gets \mathrm{Iteration}+1$\;
  }
  return failure
  \smallskip
\end{minipage}}
\end{minipage} &
\begin{minipage}{0.53\textwidth}
\includegraphics[width=\textwidth]{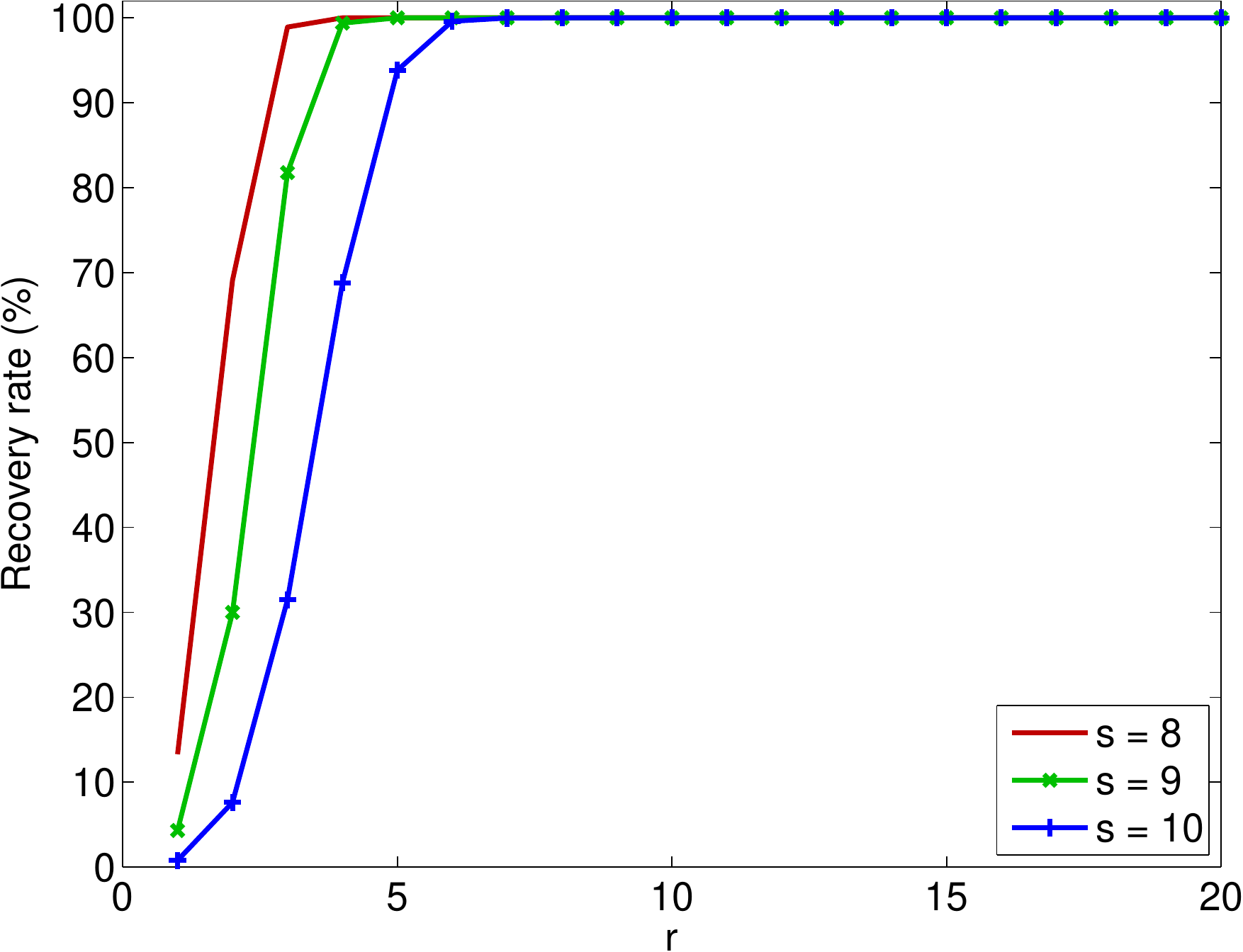}
\end{minipage} \\
\begin{minipage}[t]{0.44\textwidth}
\caption{The \rembo-$\ell_1$ algorithm}\label{Alg:ReMBo}
\end{minipage}
 &
\begin{minipage}[t]{0.53\textwidth}
\caption{Theoretical performance model for \rembo{} on three problem
   instances with different sparsity levels $s$.}\label{Fig:ReMBoBound}
\end{minipage} \\
\end{tabular}
\end{figure}

To derive the performance analysis of \rembo, we fix a support
$\mathcal{I}$ of cardinality $s$, and consider only signals with
nonzero entries on this support. Each time we multiply $B$ by a weight
vector $w$, we in fact create a new problem with an $s$-sparse
solution $x_0 = X_0w$ corresponding with a right-hand side $b = Bw =
AX_0w = Ax_0$. As reflected in~\eqref{Eq:Pl1AI}, recovery of $x_0$
using $\ell_1$ depends only on its support and sign pattern. Clearly,
the more sign patterns in $x_0$ that we can generate, the higher the
probability of recovery. Moreover, due to the elimination of
previously tried sign patterns, the probability of recovery goes up
with each new sign pattern (excluding negation of previous sign
patterns). The maximum number of sign patterns we can check with
boosted $\ell_1$ is the number of observations $r$.  The question thus
becomes, how many different sign patterns we can generate by taking
linear combinations of the columns in $X_0$? (We disregard the
situation where elimination occurs and $\card{\support(X_0w)} < s$.)
Equivalently, we can ask how many orthants in $\Real^s$ (each one
corresponding to a different sign pattern) can be properly intersected
by the hyperplane given by the range of the $s\times r$ matrix
$\bar{X}$ consisting of the nonzero rows of $X_0$ (with proper we mean
intersection of the interior). In Section~\ref{Sec:Cnd} we derive an
exact expression for the maximum number of proper orthant
intersections in $\Real^n$ by a hyperplane generated by $d$ vectors,
denoted by $C(n,d)$.

Based on the above reasoning, a good model for the recovery rate of
$n\times r$ matrices $X_0$ with $\support_{\mathrm{row}}(X_0) =
\mathcal{I} < m/2$ using \rembo{} is given by
\begin{equation}\label{Eq:PReMBo}
P\subr(A,\mathcal{I},r)
= 1 - \prod_{i=1}^{C(\card{\mathcal{I}},r)/2}\left[
  1 - \frac{\mathcal{F}_{\mathcal{I}}(A\Cscr)}%
           {\mathcal{F}_{\mathcal{I}}(\Cscr)-2(i-1)}
  \right].
\end{equation}
The term within brackets denotes the probability of failure and the
fraction represents the success rate, which is given by the ratio of
the number of faces $\mathcal{F}_{\mathcal{I}}(A\Cscr)$ that
survived the mapping to the total number of faces to consider. The
total number reduces by two at each trial because we can exclude the
face $f$ we just tried, as well as $-f$. The factor of two in
$C(\card{\mathcal{I}},r)/2$ is also due to this
symmetry\footnote{Henceforth we use the convention that the uniqueness
  of a sign pattern is invariant under negation.}.

This model would be a bound for the average performance of \rembo{} if
the sign patterns generated would be randomly sampled from the space
of all sign patterns on the given support. However, because it is
generated from the orthant intersections with a hyperplane, the actual
pattern is highly structured. Indeed, it is possible to imagine a
situation where the $(s-1)$-faces in $\Cscr$ that perish in the
mapping to $A\Cscr$ have sign patterns that are all contained in
the set generated by a single hyperplane. Any other set of sign
patterns would then necessarily include some faces that survive the
mapping and by trying all patterns in that set we would recover
$X_0$. In this case, the average recovery over all $X_0$ on that
support could be much higher than that given by~\eqref{Eq:PReMBo}. We
do not yet fully understand how the surviving faces of $\Cscr$
are distributed. Due to the simplicial structure of the facets of
$\Cscr$, we can expect the faces that perish to be partially
clustered (if a $(d-2)$-face perishes, then so will the two
$(d-1)$-faces whose intersection gives this face), and partially
unclustered (the faces that perish while all their sub-faces
survive). Note that, regardless of these patterns, recovery is
guaranteed in the limit whenever the number of unique sign patterns
tried exceeds half the number of faces lost,
$(\card{\mathcal{F}_{\mathcal{I}}(\Cscr)} -
\card{\mathcal{F}_{\mathcal{I}}(\mathcal{AC})})/2$.

Figure~\ref{Fig:ReMBoBound} illustrates the theoretical performance
model based on $C(n,d)$, for which we derive the exact expression
in Section~\ref{Sec:Cnd}. In Section~\ref{Sec:ReMBoLtd} we discuss
practical limitations, and in Section~\ref{Sec:ReMBoExperiments} we
empirically look at how the number of sign patterns generated grows
with the number of normally distributed vectors $w$, and how this
affects the recovery rates. To allow comparison between \rembo{} and
boosted $\ell_1$, we used the same matrix $A$ and support
$\mathcal{I}_s$ used to generate Figure~\ref{Fig:BoostedL1}.

\subsection{Maximum number of orthant intersections with subspace}\label{Sec:Cnd}

\begin{theorem}
  Let $C(n,d)$ denote the maximum attainable number of orthant
  interiors intersected by a hyperplane in $\Real^n$ generated by
  $d$ vectors. Then $C(n,1) = 2$, $C(n,d) = 2^n$ for $d \geq n$. In
  general, $C(n,d)$ is given by
\begin{equation}\label{Eq:Cnd}
C(n,d) = C(n-1,d-1) + C(n-1,d) = 2 \sum_{i=0}^{d-1}{n-1 \choose i}.
\end{equation}
\end{theorem}
\begin{proof}
The number of intersected orthants is exactly equal to the number of
proper sign patterns (excluding zero values) that can be generated by
linear combinations of those $d$ vectors. When $d=1$, there can only be
two such sign patterns corresponding to positive and negative
multiples of that vector, thus giving $C(n,1) = 2$. Whenever $d \geq
n$, we can choose a basis for $\Real^n$ and add additional vectors
as needed, and we can reach all points, and therefore all $2^n =
C(n,d)$ sign patterns.

For the general case~\eqref{Eq:Cnd}, let $v_1,\ldots,v_d$ be vectors
in $\Real^n$ such that the affine hull with the origin, $S =
\mathrm{aff}\{0,v_1,\ldots,v_d\}$, gives a hyperplane in
$\Real^n$ that properly intersects the maximum number of
orthants, $C(n,d)$. Without loss of generality assume that vectors
$v_i$, $i=1,\ldots,d-1$ all have their $n$th component equal to
zero. Now, let $T = \mathrm{aff}\{0,v_1,\ldots,v_{d-1}\} \subseteq
\Real^{n-1}$ be the intersection of $S$ with the
$(n-1)$-dimensional subspace of all points $\mathcal{X} =
\{x\in\Real^n \mid x_n = 0\}$, and let $C_T$ denote the number of
$(n-1)$-orthants intersected by $T$. Note that $T$ itself, as embedded
in $\Real^n$, does not properly intersect any orthant. However,
by adding or subtracting an arbitrarily small amount of $v_d$, we
intersect $2C_T$ orthants; taking $v_d$ to be the $n$th column of the
identity matrix would suffice for that matter. Any other orthants that
are added have either $x_n > 0$ or $x_n < 0$, and their number does
not depend on the magnitude of the $n$th entry of $v_d$, provided it
remains nonzero. Because only the first $n-1$ entries of $v_d$
determine the maximum number of additional orthants, the problem
reduces to $\Real^{n-1}$. In fact, we ask how many new orthants
can be added to $C_T$ taking the affine hull of $T$ with $v$, the
orthogonal projection $v_d$ onto $\mathcal{X}$. Since the maximum
orthants for this $d$-dimensional subspace in $\Real^{n-1}$ is
given by $C(n-1,d)$, this number is clearly bounded by
$C(n-1,d)-C_T$. Adding this to $2C_T$, we have
\begin{equation}\label{Eq:Cupper}
\begin{aligned}
C(n,d) & \leq 2 C_T + [C(n-1,d) - C_T] = C_T + C(n-1,d)
\\     & \leq C(n-1,d-1) + C(n-1,d) 
\\     & \leq 2\sum_{i=0}^{d-1}{n-1 \choose i}.
\end{aligned}
\end{equation}
The final expression follows by expanding the recurrence relations,
which generates (a part of) Pascal's triangle, and combining this with
$C(1,j) = 2$ for $j\geq 1$. In the above, whenever there are free
orthants in $\Real^{n-1}$, that is, when $d < n$, we can always
choose the corresponding part of $v_d$ in that orthant. As a
consequence we have that no hyperplane supported by a set of vectors
can intersect the maximum number of orthants when the range of those
vectors includes some $e_i$.

We now show that this expression holds with equality. Let $U$ denote
an $(n-d)$-hyperplane in $\Real^n$ that intersects the maximum
$C(n,n-d)$ orthants. We now claim that in the interior of each orthant
not intersected by $U$ there exists a vector that is orthogonal to
$U$. If this were not the case then $T$ must be aligned with some
$e_i$ and can therefore not be optimal. The span of these orthogonal
vectors generates a $d$-hyperplane $V$ that intersects $C_V = 2^n -
C(n,n-d)$ orthants, and it follows that
\begin{align*}
C(n,d)
& \geq C_V = 2^n - C(n,n-d) \\
& \geq 2^n - 2\sum_{i=0}^{n-d-1}{n-1 \choose i}
  = 2\sum_{i=0}^{n-1}{n-1 \choose i} - 2\sum_{i=0}^{n-d-1}{n-1 \choose i} \\
& = 2\sum_{n-d}^{n-1}{n-1 \choose i} = 2\sum_{i=0}^{d-1}{n-1 \choose i} 
  \geq C(n,d),
\end{align*}
where the last inequality follows from \eqref{Eq:Cupper}.
Consequently, all inequalities hold with equality.
\end{proof}

\begin{corollary}
  Given $d \leq n$, then $C(n,d) = 2^n - C(n,n-d)$, and $C(2d,d) = 2^{2d-1}$.
\end{corollary}


\begin{corollary}
  A hyperplane $\mathcal{H}$ in $\Real^n$, defined as the
  range of $V = [v_1,\ v_2, \ldots,\ v_d]$, intersects the maximum number
  of orthants $C(n,d)$ whenever $\rank(V) = n$, or when $e_i
  \not\in\range(V)$ for $i=1,\ldots,n$.
\end{corollary}

\subsection{Practical considerations}\label{Sec:ReMBoLtd}

In practice it is generally not feasible to generate all of the
$C(\card{\mathcal{I}},r)/2$ unique sign patterns. This means that we
would have to replace this term in~\eqref{Eq:PReMBo} by the number of
unique patterns actually tried.  For a given $X_0$ the actual
probability of recovery is determined by a number of factors. First of
all, the linear combinations of the columns of the nonzero part of
$\bar{X}$ prescribe a hyperplane and therefore a set of possible sign
patterns. With each sign pattern is associated a face in $\Cscr$
that may or may not map to a face in $A\Cscr$. In addition,
depending on the probability distribution from which the weight
vectors $w$ are drawn, there is a certain probability for reaching
each sign pattern. Summing the probability of reaching those patterns
that can be recovered gives the probability $P(A,\mathcal{I},X_0)$ of
recovering with an individual random sample $w$. The probability of
recovery after $t$ trials is then of the form
\[
1 - [1 - P(A,\mathcal{I},X_0)]^t.
\]
To attain a certain sign pattern $\ebar$, we need to find an
$r$-vector $w$ such that $\sign(\bar{X}w) = \ebar$. For a positive
sign on the $j$th position of the support we can take any vector $w$
in the open halfspace $\{w \mid \bar{X}\row{j}w > 0\}$, and likewise
for negative signs. The region of vectors $w$ in $\Real^r$ that
generates a desired sign pattern thus corresponds to the intersection
of $\abs{\Iscr}$ open halfspaces. The measure of this intersection as
a fraction of $\Real^r$ determines the probability of sampling such a
$w$.  To formalize, define $\mathcal{K}$ as the cone generated by the
rows of $-\diag(\ebar)\bar{X}$, and the unit Euclidean $(k-1)$-sphere
$\mathcal{S}^{k-1} = \{x\in\Real^r \mid \norm{x}_2 = 1\}$. The
intersection of halfspaces then corresponds to the interior of the
polar cone of $\mathcal{K}$: $\mathcal{K}\polar = \{x\in\Real^r \mid
x\T y \leq 0,\ \forall y\in\mathcal{K}\}$. The fraction of $\Real^r$
taken up by $\mathcal{K}\polar$ is given by the $(k-1)$-content of
$\mathcal{S}^{k-1} \cap \mathcal{K}\polar$ to the $(k-1)$-content of
$\mathcal{S}^{k-1}$~\cite{GRU2003a}. This quantity coincides precisely
with the definition of the external angle of $\mathcal{K}$ at the
origin.

\subsection{Experiments}\label{Sec:ReMBoExperiments}

In this section we illustrate the theoretical results from
Section~\ref{Sec:ReMBo} and examine some practical considerations that
affect the performance of \rembo. For all experiments that require the
matrix $A$, we use the same $20\times 80$ matrix that was used in
Section~\ref{Sec:BoostedL1}, and likewise for the supports
$\mathcal{I}_s$.  To solve~\eqref{Eq:L1}, we again use CVX in
conjunction with SDPT3. We consider $x_0$ to be recovered from $b =
Ax_0 = AX_0w$ if $\norm{x^*-x_0}_\infty\le 10^{-5}$, where $x^*$ is the computed
solution.

The experiments that are concerned with the number of unique sign
patterns generated depend only on the $s\times r$ matrix $\bar{X}$
representing the nonzero entries of $X_0$. Because an initial
reordering of the rows does not affect the number of patterns, those
experiments depend only on $\bar{X}$, $s = \card{\mathcal{I}}$, and
the number of observations $r$; the exact indices in the support
set $\mathcal{I}$ are irrelevant for those tests.


\subsubsection{Generation of unique sign patterns}

The practical performance of \rembo{} depends on its ability to
generate as many different sign patterns using the columns in $X_0$ as
possible. A natural question to ask then is how the number of such
patterns grows with the number of randomly drawn samples $w$. Although
this ultimately depends on the distribution used for generating the
entries in $w$, we shall, for sake of simplicity, consider only
samples drawn from the normal distribution. As an experiment we take a
$10\times 5$ matrix $\bar{X}$ with normally-distributed entries, and
over $10^8$ trials record how often each sign-pattern (or negation)
was reached, and in which trial they were first encountered. The
results of this experiment are summarized in
Figure~\ref{Fig:ApproachC10_5}. From the distribution in
Figure~\ref{Fig:ApproachC10_5}(b) it is clear that the occurrence
levels of different orthants exhibits a strong bias.
The most frequently visited orthant pairs were reached up to
$7.3\times 10^6$ times, while others, those hard to reach using
weights from the normal distribution, were observed only four times
over all trials. The efficiency of \rembo{} depends on the rate of
encountering new sign patterns.  Figure~\ref{Fig:ApproachC10_5}(c) shows
how the average rate changes over the number of trials. The curves in
Figure~\ref{Fig:ApproachC10_5}(d) illustrate the theoretical probability
of recovery in~\eqref{Eq:PReMBo}, with $C(n,d)/2$ replaced by the
number of orthant pairs at a given iteration, and with face counts
determined as in Section~\ref{Sec:BoostedL1}, for three instances with
support cardinality $s=10$, and observations $r=5$.

\begin{figure}
\centering
\begin{tabular}{cc}
\includegraphics[width=0.45\textwidth]{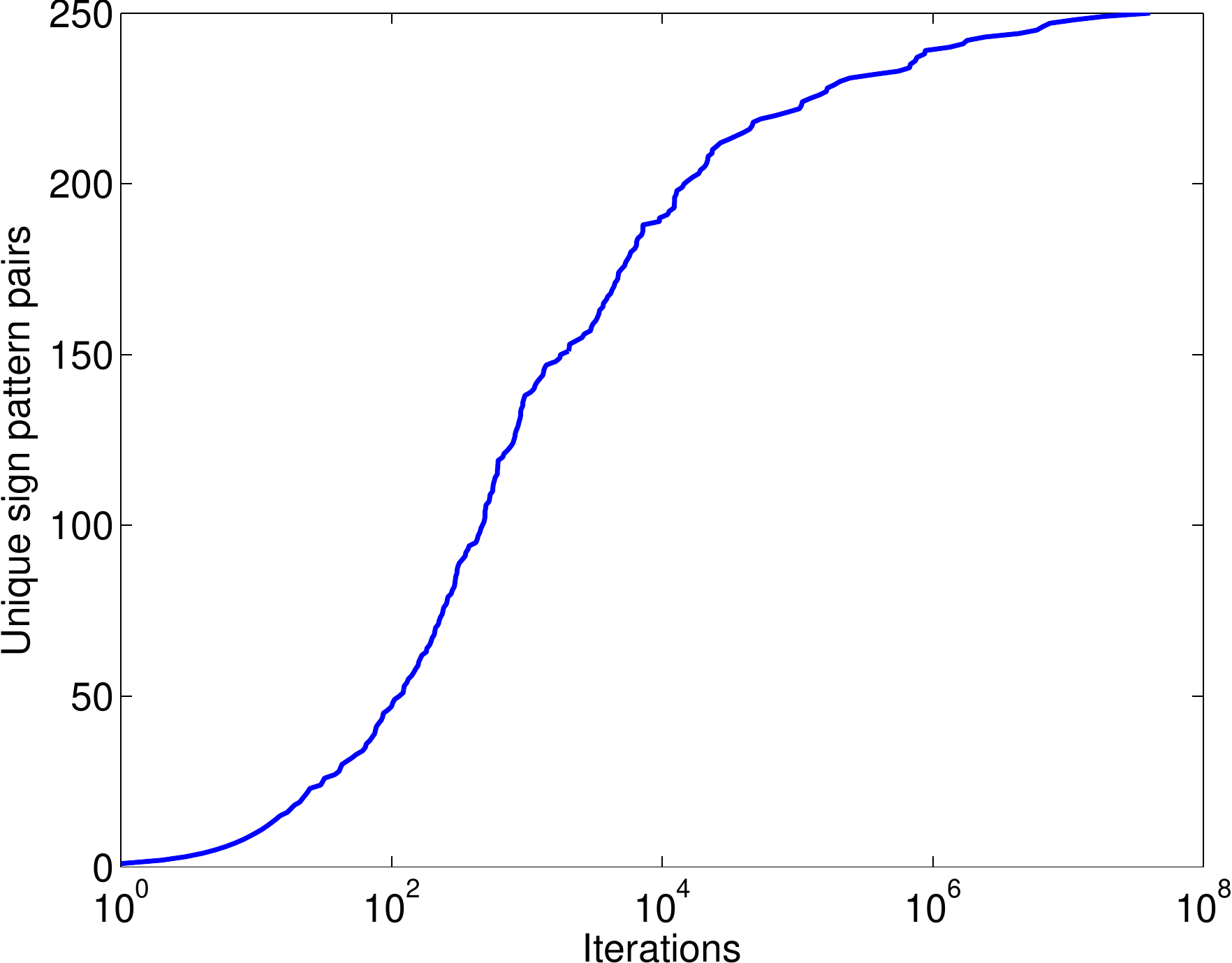} &
\includegraphics[width=0.45\textwidth]{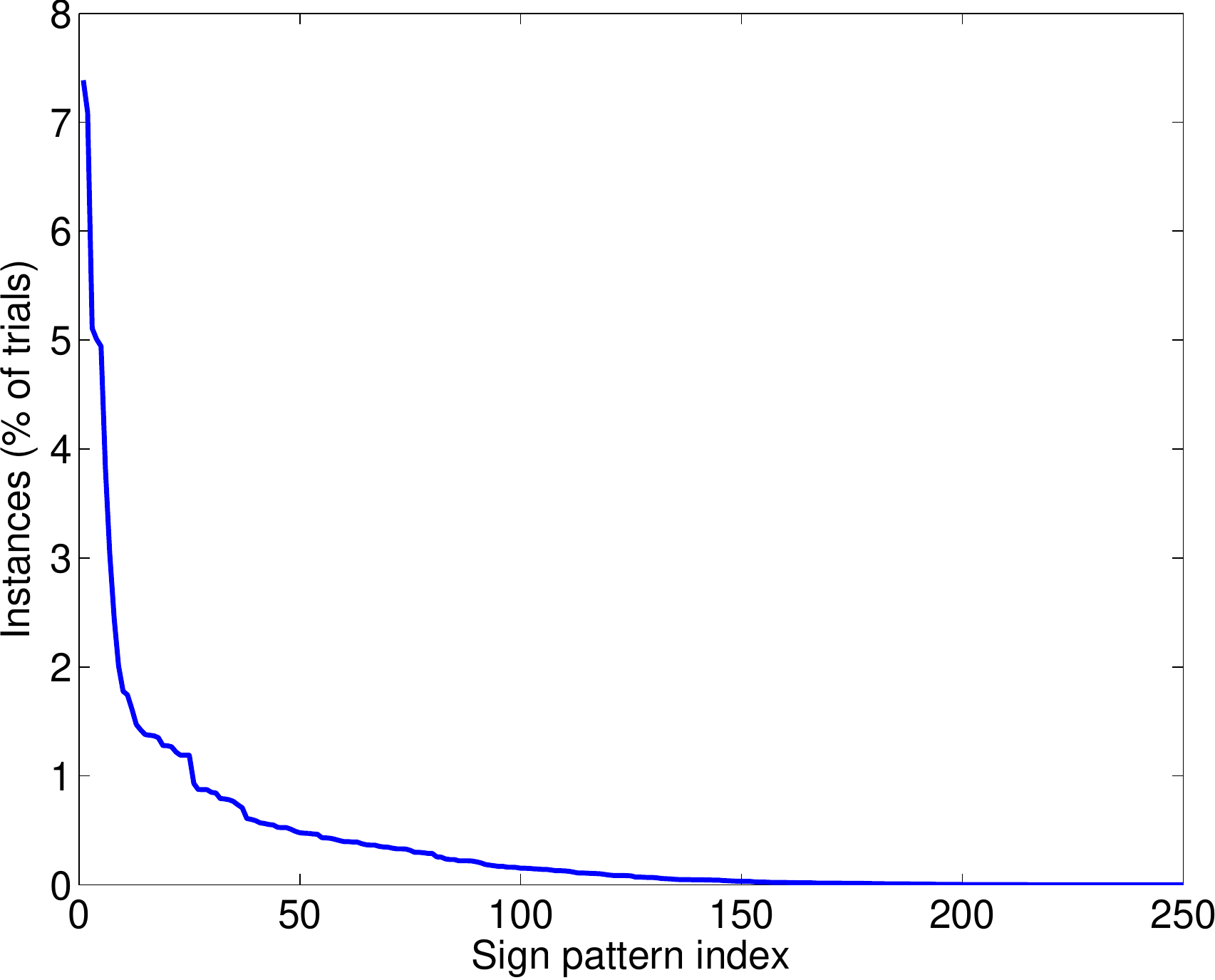}
\\ (a) & (b) \\[6pt]
\includegraphics[width=0.45\textwidth]{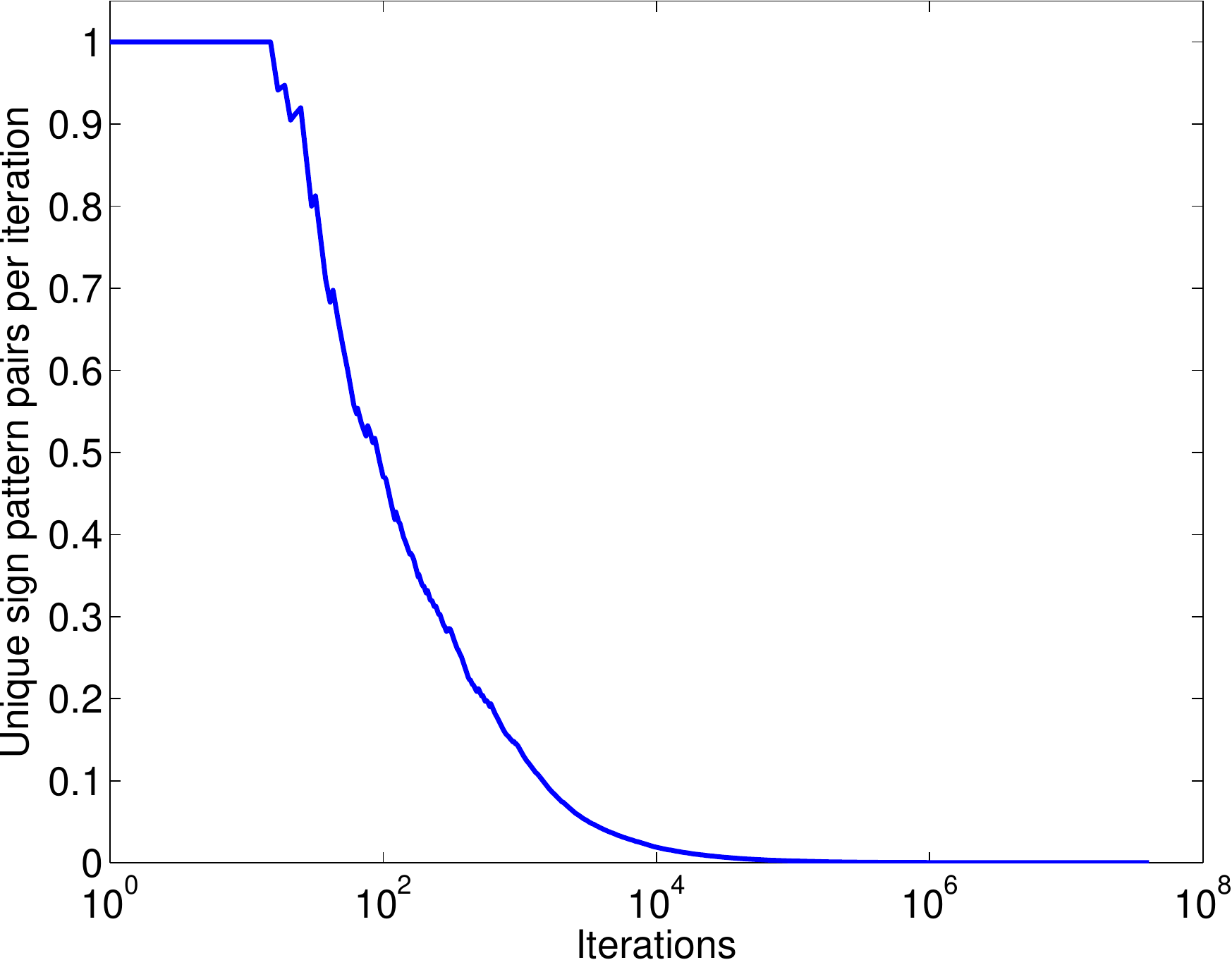} &
\includegraphics[width=0.45\textwidth]{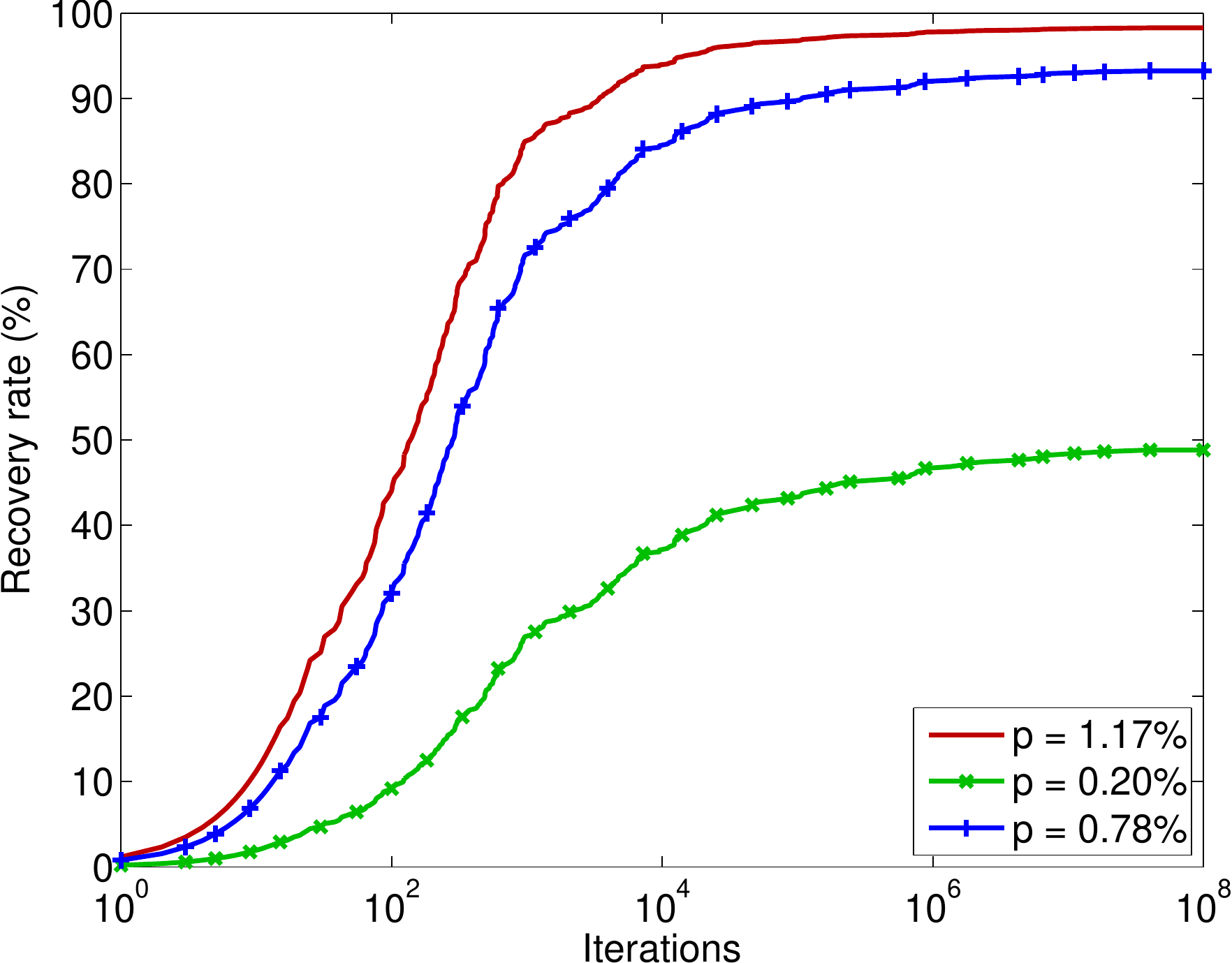} \\
(c) & (d)
\end{tabular}
\caption{Sampling the sign patterns for a $10\times 5$ matrix
  $\bar{X}$, with (a) number of unique sign patterns versus number of
  trials, (b) relative frequency with which each orthant is sampled,
  (c) average number of new sign patterns per iteration as a function
  of iterations, and (d) theoretical probability of recovery using
  \rembo{} for three instances of $X_0$ with row sparsity $s = 10$, and
  $r=5$ observations.}
\label{Fig:ApproachC10_5}
\end{figure}

\subsubsection{Role of $\bar{X}$.}
Although the number of orthants that a hyperplane can intersect does
not depend on the basis with which it was generated, this choice does
greatly influence the ability to sample those
orthants. Figure~\ref{Fig:MatrixInfluenceApproachCnd} shows two ways
in which this can happen. In part (a) we sampled the number of unique
sign patterns for two different $9\times 5$ matrices $\bar{X}$, each
with columns scaled to unit $\ell_2$-norm. The entries of the first
matrix were independently drawn from the normal distribution, while
those in the second were generated by repeating a single column drawn
likewise and adding small random perturbations to each entry. This
caused the average angle between any pair of columns to decrease from
$65$ degrees in the random matrix to a mere $8$ in the perturbed
matrix, and greatly reduces the probability of reaching certain
orthants. The same idea applies to the case where $d \geq n$, as shown
in part (b) of the same figure. Although choosing $d$ greater than $n$
does not increase the number of orthants that can be reached, it does
make reaching them easier, thus allowing \rembo{} to work more
efficiently. Hence, we can expect \rembo{} to have higher recovery on
average when the number of columns in $X_0$ increases and when they
have a lower mutual coherence $\mu(X) = \min_{i\neq j}\abs{x_i^T
  x_j}/(\norm{x_i}_2\cdot\norm{x_j}_2)$.

\begin{figure}
\centering
\begin{tabular}{cc}
\includegraphics[width=0.45\textwidth]{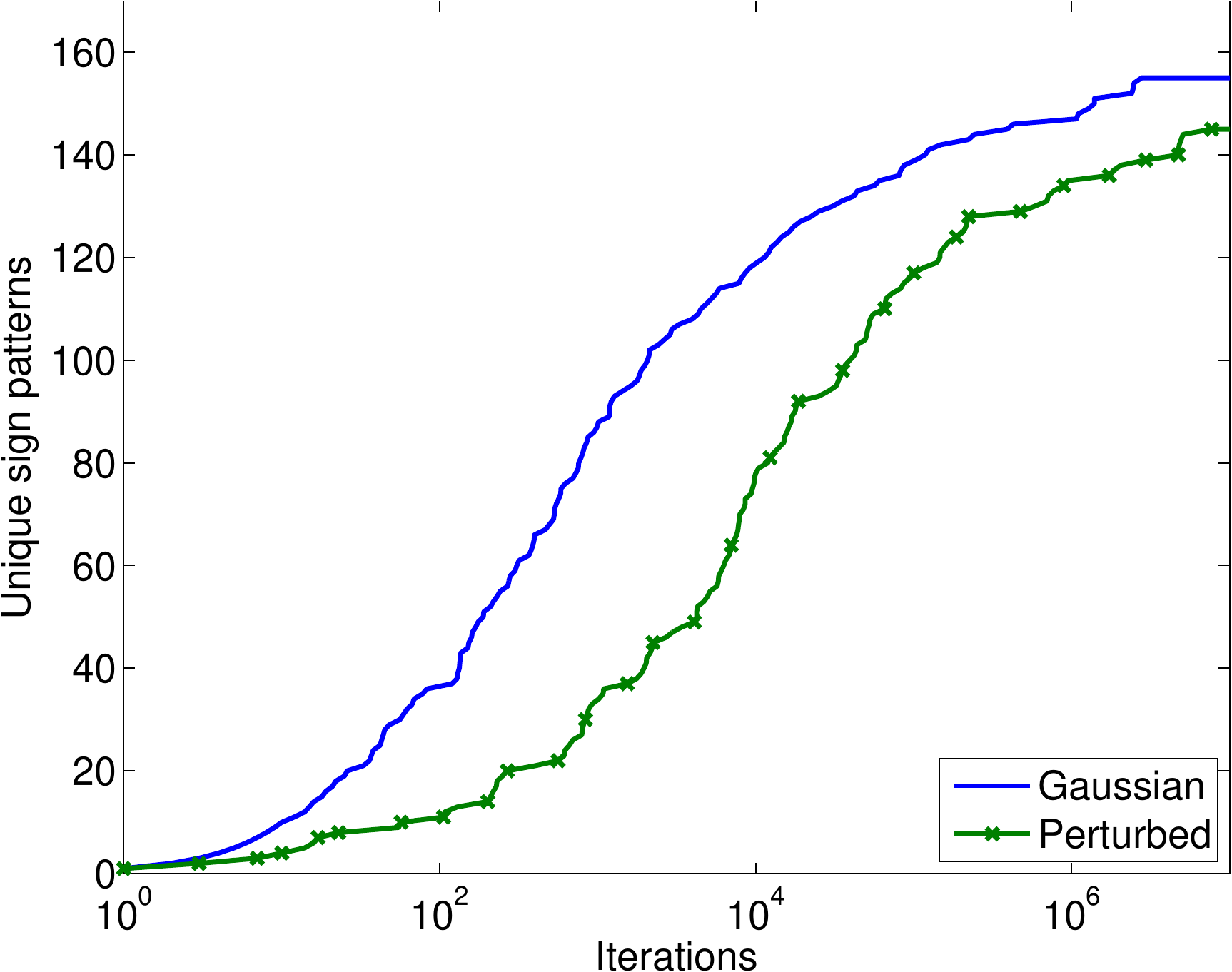} &
\includegraphics[width=0.45\textwidth]{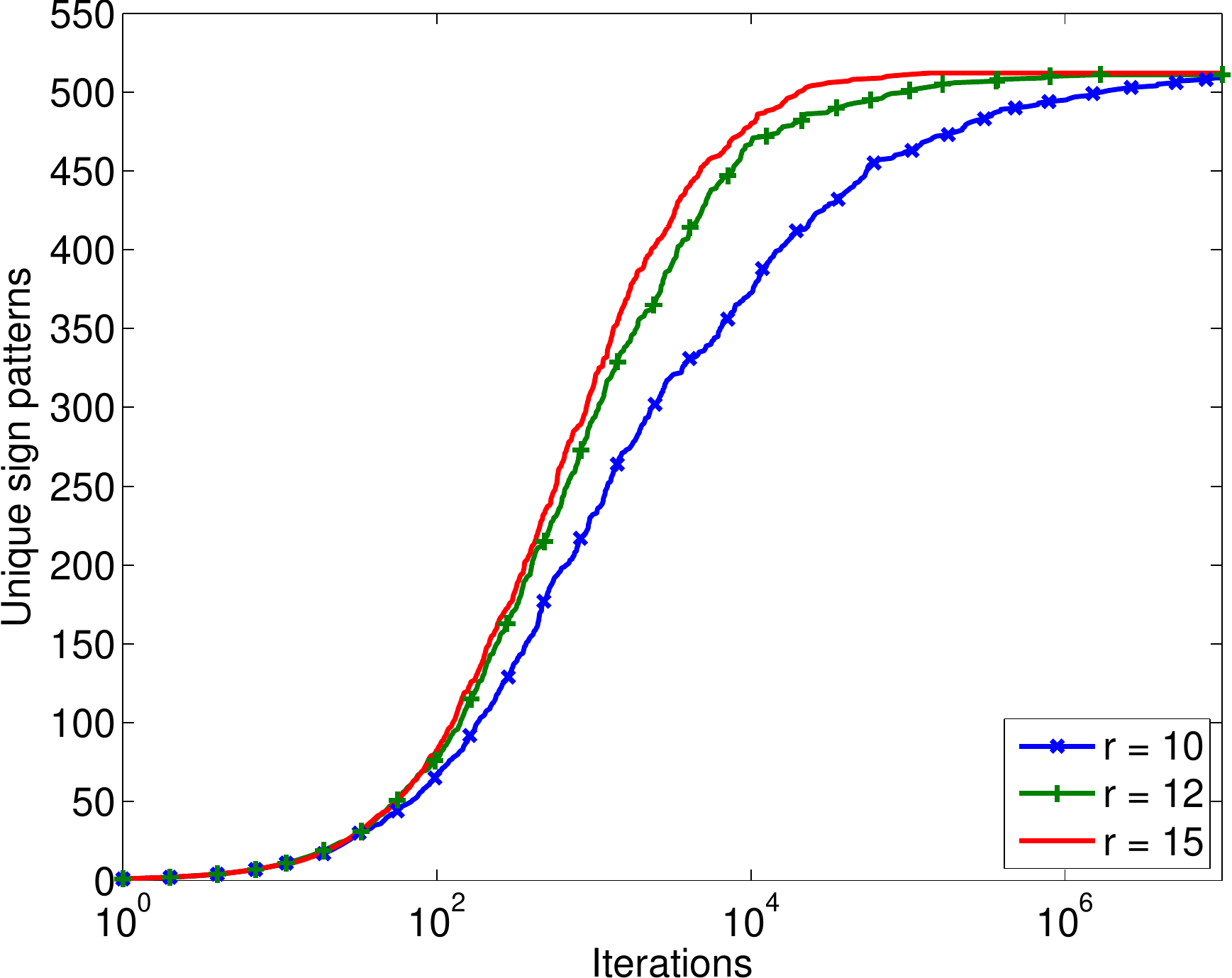} \\
(a) & (b)
\end{tabular}
\caption{Number of unique sign patterns for (a) two $9\times 5$
  matrices $\bar{X}$ with columns scaled to unit $\ell_2$-norm; one
  with entries drawn independently from the normal distribution, and
  one with a single random column repeated and random perturbations
  added, and (b) $10\times r$ matrices with $r=10,12,15$.}
\label{Fig:MatrixInfluenceApproachCnd}
\end{figure}

\subsubsection{Limiting the number of iterations}

The number of iterations used in the previous experiments greatly
exceeds that what is practically feasible: we cannot afford to run
\rembo{} until all possible sign patterns have been tried, even if there
was a way detect that the limit had been reached. Realistically, we
should set the number of iterations to a fixed maximum that depends
on the computational resources available, and the problem setting.

In Figure~\ref{Fig:ApproachC10_5} we show the unique orthant count as
a function of iterations and the predicted recovery rate. When using
only a limited number of iterations it is interesting to know what the
distribution of unique orthant counts looks like. To find out, we drew
1,000 random $\bar{X}$ matrices for each size $s\times r$, with $s=10$
nonzero rows fixed, and the number of columns ranging from
$r=1,\ldots,20$. For each $\bar{X}$ we counted the number of unique
sign patterns attained after respectively 1,000 and 10,000
iterations. The resulting minimum, maximum, and median values are
plotted in Figure~\ref{Fig:ReMBoPerformance}(a) along with the
theoretical maximum. More interestingly of course is the average
recovery rate of \rembo{} with those number of iterations. For this test
we again used the $20\times 80$ matrix $A$ with predetermined support
$\mathcal{I}$, and with success or failure of each sign pattern on
that support precomputed. For each value of $r=1,\ldots,20$ we
generated random matrices $X$ on $\mathcal{I}$ and ran \rembo{} with the
maximum number of iterations set to 1,000 and 10,000. To save on
computing time, we compared the on-support sign pattern of each combined
coefficient vector $Xw$ to the known results instead of solving
$\ell_1$. The average recovery rate thus obtained is plotted in
Figures~\ref{Fig:ReMBoPerformance}(b)--(c), along with the average of the
predicted performance using~\eqref{Eq:PReMBo} with $C(n,d)/2$ replaced
by orthant counts found in the previous experiment.

\begin{figure}
\centering
\begin{tabular}{ccc}
\includegraphics[width=0.32\textwidth]{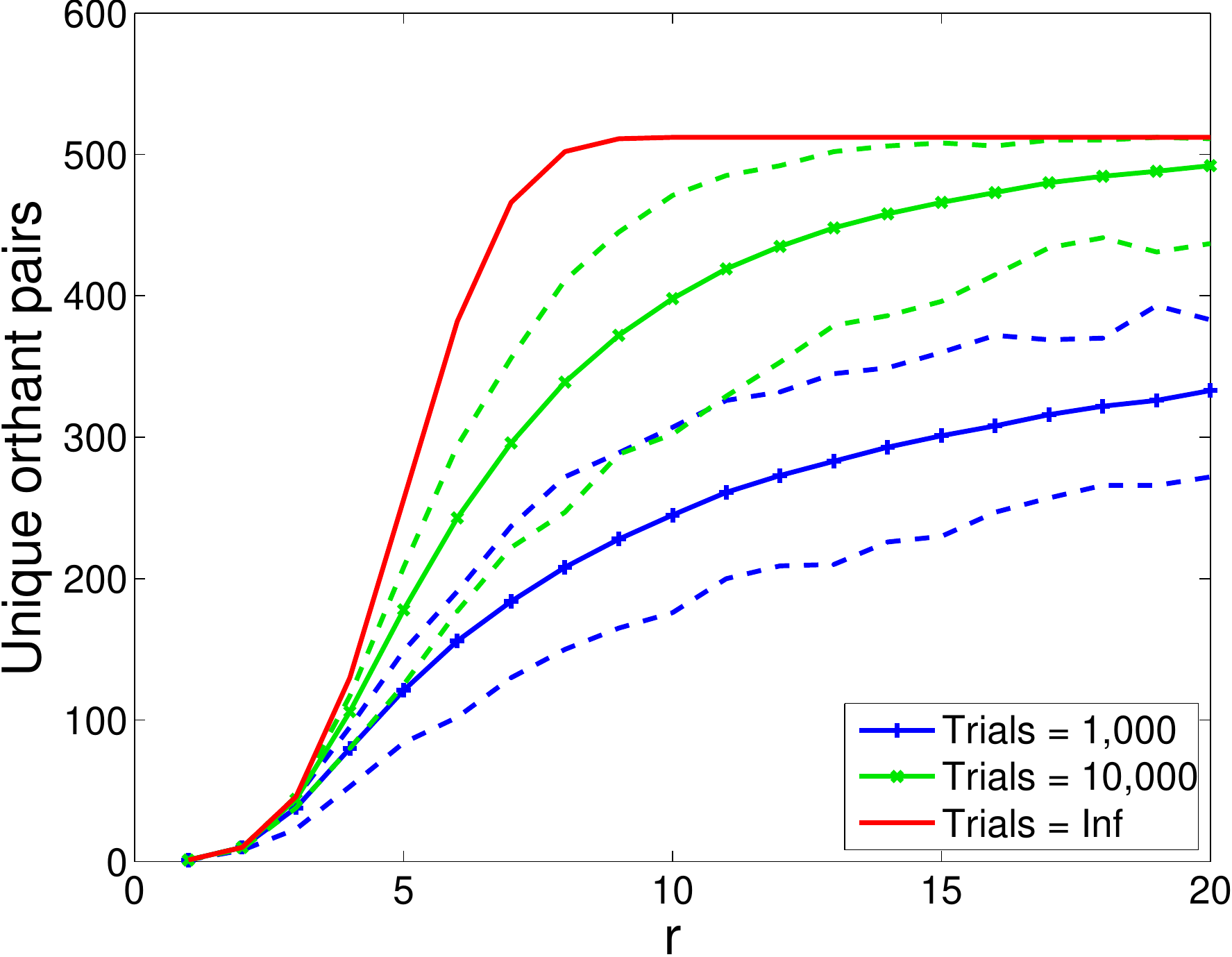} &
\includegraphics[width=0.32\textwidth]{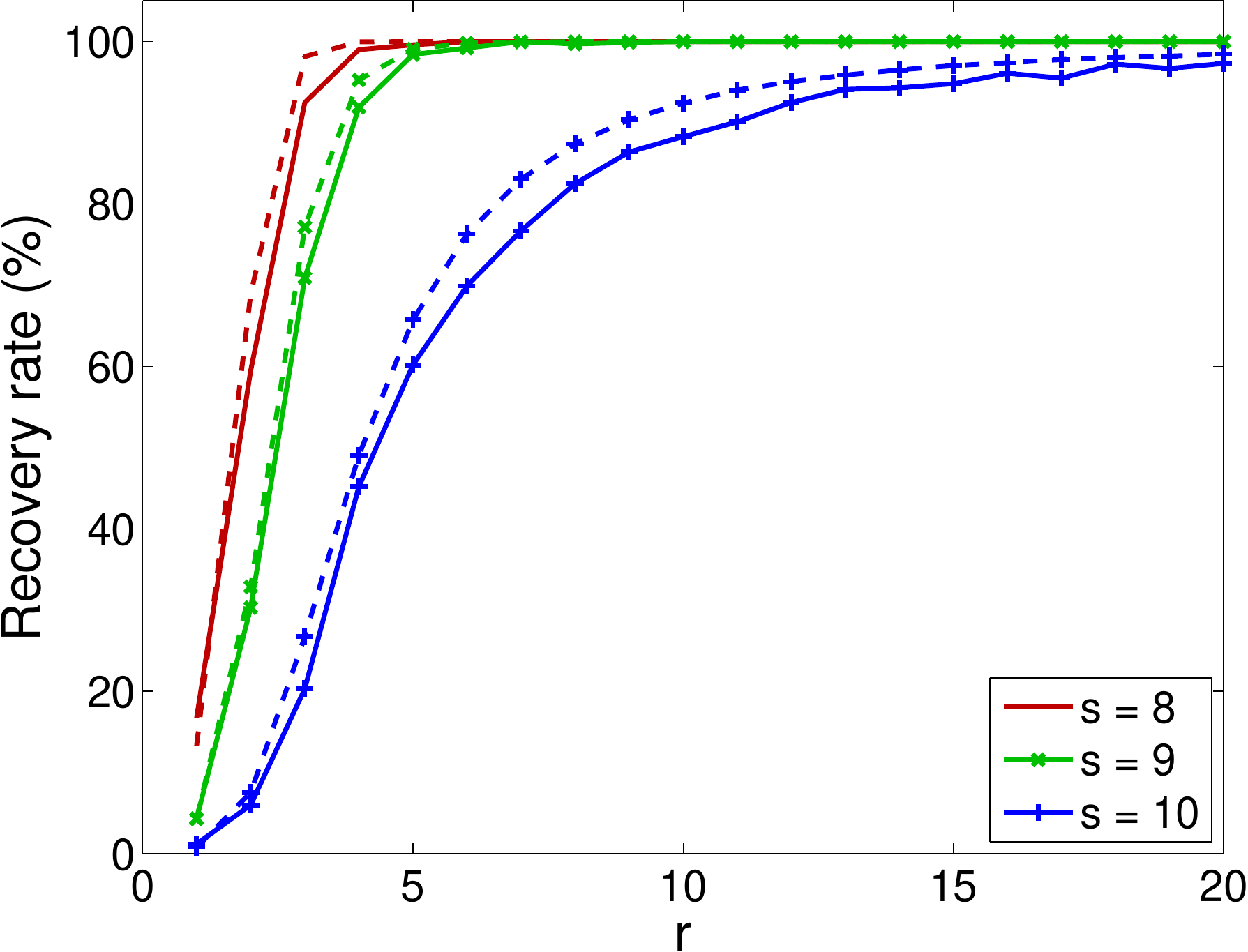} &
\includegraphics[width=0.32\textwidth]{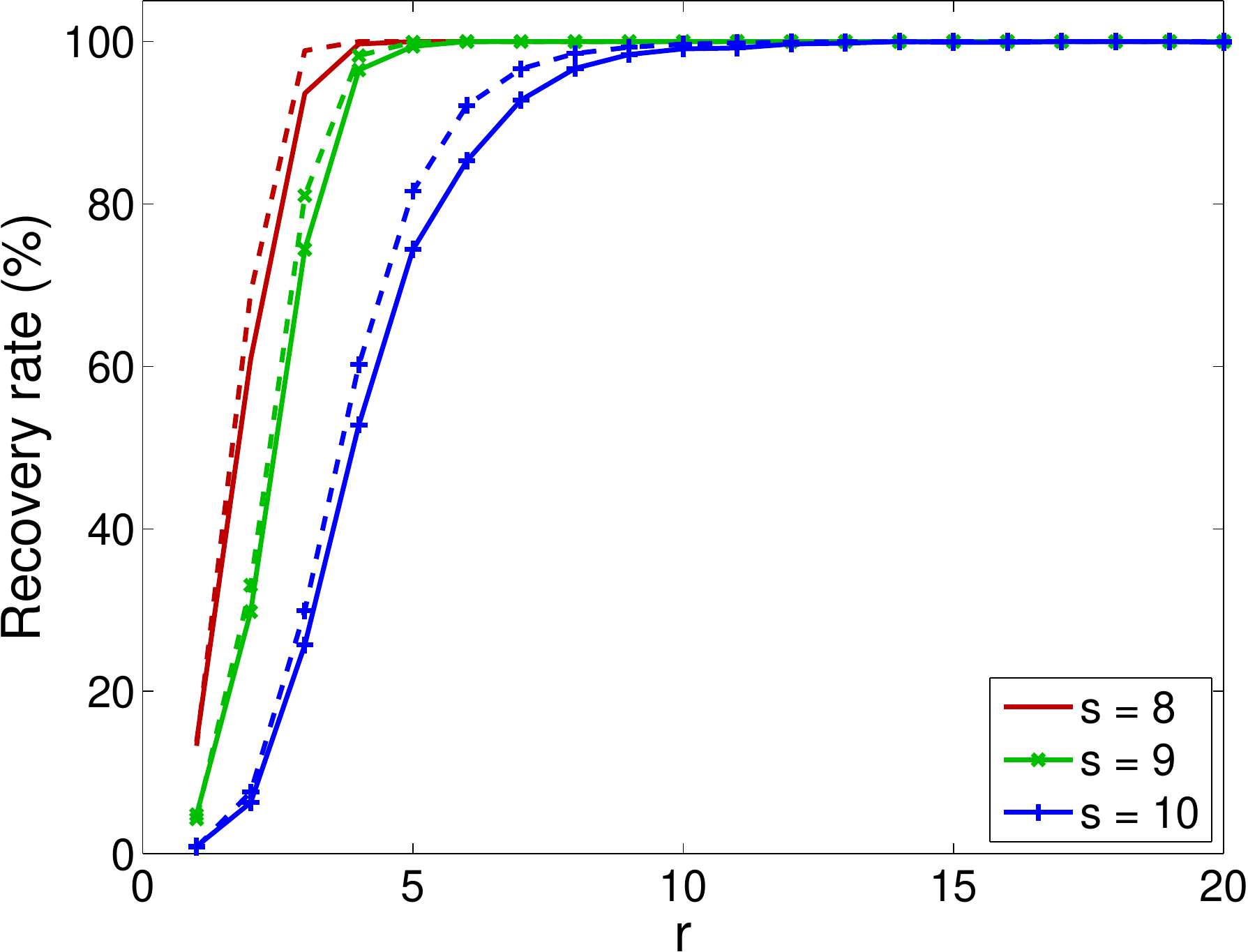}
\\
(a) & (b) & (c)
\end{tabular}
\caption{Effect of limiting the number of weight vectors $w$ on (a)
  the distribution of unique orthant counts for $10\times k$ random
  matrices  $\bar{X}$, solid lines give the median number and the
  dashed lines indicate the minimum and maximum values, the top solid
  line is the theoretical maximum; (b--c) the average performance of the
  \rembo-$\ell_1$ algorithm (solid) for fixed $20\times 80$ matrix $A$
  and three different support sizes $r=8,9,10$, along with the average
  predicted performance (dashed). The support patterns used are the
  same as those used for Figure \ref{Fig:BoostedL1}.}\label{Fig:ReMBoPerformance}
\end{figure}

\section{Conclusions}

The MMV problem is often solved by minimizing the sum-of-row norms of
the unknown coefficients $X$. We show that the (local) uniform
recovery properties, i.e., recovery of all $X_0$ with a fixed row
support $\mathcal{I} = \support_{\mathrm{row}}(X_0)$, cannot exceed
that of $\ell_{1,1}$, the sum of $\ell_1$ norms. This is despite the
fact that $\ell_{1,1}$ reduces to solving the basis pursuit problem
\eqref{Eq:L1} for each column separately, which does not take
advantage of the fact that all vectors in $X_0$ are assumed to have
the same support. A consequence of this observation is that the use of
restricted isometry techniques to analyze (local) uniform recovery
using sum-of-norm minimization can at best give improved bounds on
$\ell_1$ recovery.

Empirically, minimization with $\ell_{1,2}$, the sum of $\ell_2$
norms, clearly outperforms $\ell_{1,1}$ on individual problem
instances: for supports where uniform recovery fails, $\ell_{1,2}$
recovers more cases than $\ell_{1,1}$. We construct cases where
$\ell_{1,2}$ succeeds while $\ell_{1,1}$ fails, and vice versa. From
the construction where only $\ell_{1,2}$ succeeds it also follows that
the relative magnitudes of the coefficients in $X_0$ matter for
recovery. This is unlike $\ell_{1,1}$ recovery, where only the support
and the sign patterns matter. This implies that the notion of faces,
so useful in the analysis of $\ell_1$, disappears.

We show that the performance of $\ell_{1,1}$ outside the
uniform-recovery regime degrades rapidly as the number of observations
increases. We can turn this situation around, and increase the
performance with the number of observations by using a
boosted-$\ell_1$ approach. This technique aims to uncover the correct
support based on basis pursuit solutions for individual observations.
Boosted-$\ell_1$ is a special case of the \rembo{} algorithm which
repeatedly takes random combinations of the observations, allowing it
to sample many more sign patterns in the coefficient space. As a
result, the potential recovery rates of \rembo{} (at least in
combination with an $\ell_1$ solver) are a much higher than
boosted-$\ell_1$. \rembo{} can be used in combination with any solver
for the single measurement problem $Ax=b$, including greedy approaches
and reweighted $\ell_1$ \cite{CAN2008WBa}. The recovery rate of greedy
approaches may be lower than $\ell_1$ but the algorithms are generally
much faster, thus giving \rembo{} the chance to sample more random
combinations. Another advantage of \rembo, even more so than
boosted-$\ell_1$, is that it can be easily parallelized.

Based on the geometrical interpretation of \rembo-$\ell_1$ (cf.\@
Figure~\ref{Alg:ReMBo}), we conclude that, theoretically, its
performance does not increase with the number of observations after
this number reaches the number of nonzero rows.  In addition we
develop a simplified model for the performance of \rembo-$\ell_1$.  To
improve the model we would need to know the distribution of faces in
the cross-polytope $\Cscr$ that map to faces on $A\Cscr$, and the
distribution of external angles for the cones generated by the signed
rows of the nonzero part of $X_0$.

It would be very interesting to compare the recovery performance
between $\ell_{1,2}$ and \rembo-$\ell_1$. However, we consider this
beyond the scope of this paper.

All of the numerical experiments in this paper are reproducible. The
scripts used to run the experiments and generate the figures can be
downloaded from
\begin{center}\url{http://www.cs.ubc.ca/~mpf/jointsparse}.\end{center}

\section*{Acknowledgments}

The authors would like to give their sincere thanks to \"{O}zg\"{u}r
Y{\i}lmaz and Rayan Saab for their thoughtful comments and suggestions
during numerous discussions.

\bibliography{../bibliography/articles}

\end{document}